\documentclass[11pt]{article}

\usepackage{fullpage}

\usepackage{balance} %

\title{How to Learn Real-Life Approval Elections?}
\title{Learning Real-Life Approval Elections\thanks{For our code and data, see \url{https://github.com/Project-PRAGMA/Learning-IAMs-AAMAS-2025}}}

\author{
\begin{tabular}{cc}
\begin{tabular}{c}
Piotr Faliszewski\\
AGH University of Kraków\\
Poland
\end{tabular}
&
\begin{tabular}{c}
Łukasz Janeczko\\
AGH University of Kraków\\
Poland
\end{tabular}
  \\[2em]
\begin{tabular}{c}
Andrzej Kaczmarczyk\\
Department of Computer Science\\
The University of Chicago\\
USA
\end{tabular}
&
\begin{tabular}{c}
Marcin Kurdziel\\
AGH University of Kraków\\
Poland
\end{tabular}
\\[2em]
\begin{tabular}{c}
Grzegorz Pierczyński\\
AGH University of Kraków\\
Poland
\end{tabular}
&
\begin{tabular}{c}
Stanisław Szufa\\
CNRS, LAMSADE\\
Université Paris Dauphine -- PSL\\
France
\end{tabular}
\end{tabular}
}

\usepackage{url}
\usepackage{amsmath}
\usepackage{amsfonts}
\usepackage{amsthm}
\usepackage[utf8]{inputenc}
\usepackage{caption}
\usepackage{graphicx}
\usepackage{booktabs}
\usepackage{algorithm}
\usepackage{algorithmic}
\usepackage{comment}

\usepackage[T1]{fontenc}

\newtheorem{theorem}{Theorem}[section]
\newtheorem{claim}[theorem]{Claim}
\newtheorem{proposition}[theorem]{Proposition}
\newtheorem{corollary}[theorem]{Corollary}
\newtheorem{remark}{Remark}[section]
\newtheorem{definition}{Definition}[section]
\usepackage[textsize = scriptsize] {todonotes}
\newcommand{\mknote}[1]{\todo[color=cyan!10,inline]{MK: #1}}

\newcommand{\aknote}[1]{\todo[color=red!15,inline]{AK: #1}}

\usepackage[hidelinks]{hyperref}
\usepackage{cleveref}

\usepackage{appendix}

\usepackage{subcaption}
\usepackage{nicefrac}
\usepackage{fontawesome5}
\usepackage{pifont}
\usepackage{xcolor}

\usepackage{tikz}
\usepackage{pgfplots}
\usepgfplotslibrary{groupplots}
\pgfplotsset{compat=1.17,
	legend image code/.code={
		\draw[mark repeat=2,mark phase=2]
		plot coordinates {
			(0cm,0cm)
			(0.15cm,0cm)        %
			(0.3cm,0cm)         %
		};%
}}

\newcommand{\prob}{\mathbb{P}}
\usepackage{graphicx} %
\usepackage{bm} %
\usepackage{float} %
\usepackage{pdfpages} %

\newcommand{\calA}{\mathcal{A}}

\newcommand{\calD}{\mathcal{D}}
\newcommand{\cond}{\!\mid\!}

\newcommand{\eval}{{\text{eval}}}
\newcommand{\learn}{{\text{learn}}}
\newcommand{\sample}{{\text{sample}}}
\newcommand{\try}{{\text{try}}}

\newcommand{\pr}{{{\mathrm{prob}}}}
\newcommand{\app}{{\mathrm{app}}}

\newcommand{\ham}{{{\mathrm{ham}}}}
\newcommand{\vah}{{{\mathrm{va}\text{-}\mathrm{ham}}}}
\newcommand{\Ham}{{{\mathrm{Ham}}}}

\usepackage[numbers,sort,compress]{natbib}

\begin{document}

\maketitle

\begin{abstract}
  We study the independent approval model (IAM) for approval
  elections, where each candidate has its own approval probability and
  is approved independently of the other ones. This model generalizes,
  e.g., the impartial culture, the Hamming noise model, and the
  resampling model. We propose algorithms for learning IAMs and their
  mixtures from data, using either maximum likelihood estimation or
  Bayesian learning. We then apply these algorithms
  to a large set of elections from the Pabulib database. In
  particular, we find that single-component models are rarely
  sufficient to capture the complexity of real-life data, whereas
  their mixtures perform~well.
\end{abstract}

\section{Introduction}

The goal of this paper is to design algorithms that take an approval
election as input and produce
simple probabilistic models for generating similar elections (models
of generating random elections are often called statistical cultures).
We form such algorithms using maximum likelihood estimation (MLE) and
Bayesian learning approaches, and
evaluate
them
on the data from the Pabulib collection of real-life participatory
budgeting elections~\citep{fal-fli-pet-pie-sko-sto-szu-tal:c:pabulib}.
Consequently, we also get an insight into the nature of Pabulib data.\footnote{We disregard the costs of the
  projects (candidates) present in participatory budgeting.}

More formally, an approval election consists of a set of candidates
and a collection of voters.
Each voter indicates which candidates he or she finds appealing, i.e.,
which ones he or she approves, and which ones he or she does
not. %
One particularly natural model of generating such elections is to
provide for each candidate $c$ a probability $p_c$, so that each voter
approves $c$ with this probability, independently from the other
candidates and voters.
We refer to this model as the independent
approval model (IAM).  For a positive integer $t$, by a $t$-parameter
IAM we mean an IAM where each candidate has one of at most $t$
different approval probabilities, and we often write \emph{full IAM}
when no such restriction applies.
While full IAM is not widely used in computational
social choice---indeed, it does not appear in the recent overview of
\citet{boe-fal-jan-kac-lis-pie-rey-sto-szu-was:c:guide}---its
restricted variants are quite popular.  For example, $1$-IAM---where
each candidate is approved independently with some probability
$p$---is simply the $p$-Impartial Culture model ($p$-IC), one of the
most popular statistical cultures for approval
elections~\citep{boe-fal-jan-kac-lis-pie-rey-sto-szu-was:c:guide}, and
$2$-IAM is equivalent to the resampling model of
\citet{szu-fal-jan-lac-sli-sor-tal:c:sampling-approval-elections}.
We also note that the Hamming noise model---analyzed, e.g., by
\citet{car-kak-kar-kri:j:noise-approval}, is a restricted variant of
$2$-IAM.

We provide algorithms that given an approval election and a
number~$t$, find a $t$-parameter IAM that maximizes the probability of
generating this election. These algorithms are simple for impartial
culture, Hamming noise model, and full IAM, while resampling and
general $t$-parameter IAMs require more effort. We also show how the
classic expectation-maximization (EM) and Bayesian learning algorithms
can be used to learn mixtures of $t$-parameter IAMs, albeit with
limited guarantees. In this case, we focus on the Hamming noise model,
resampling, and full IAM.

There are two main reasons why such learning algorithms are
useful. The first one is that using them we can get a strong insight
into the nature of the elections that we learn. In our case, we
consider all 271 approval elections from
Pabulib~\citep{fal-fli-pet-pie-sko-sto-szu-tal:c:pabulib}
with up to tens of thousands of voters and between a few to $200$ candidates.
For each of these elections we learn each of
our models. We find that while single IAMs are sufficient for some of
the instances, in most cases it is necessary to consider mixtures of
at least a few IAM components. In addition, we find that some elections are
 inherently difficult to learn, irrespectively how strong would
our models be.

The second reason why our learning algorithms---especially those for
mixture models---are useful, is that they provide models which
generate fairly realistic synthetic data. In this sense, such models
are more realistic than basic, stylized models often used in the
literature (see the overview
of~\citet{boe-fal-jan-kac-lis-pie-rey-sto-szu-was:c:guide}), but still
offer a strong level of control over the generated data. For example,
we can generate as many votes as we like.

\paragraph{Related Work}
So far, learning approval elections did not receive much attention in
computational social
choice~\citep{boe-fal-jan-kac-lis-pie-rey-sto-szu-was:c:guide}.
However, we mention a paper that analyzes the number of approval votes
that we need to sample to learn an underlying ground
truth~\citep{car-mic:c:learning-mallows}.
Further, \citet{rol-aub-gan-leo:j:learning-evaluation-voting} consider learning
profiles where each voter assigns a score to each candidate that
either comes from a continuous domain or from a discrete one. For the
discrete case, their setting generalizes ours, but they consider quite
different distributions and learning approaches.
On the other hand, learning
models of ordinal elections, where each voter ranks the candidates, is
well-represented. For example, there are %
algorithms for
learning the classic Mallows
model~\citep{lu-bou:j:sampling-mallows,bet-bre-nie:j:kemeny,boe-bre-elk-fal-szu:c:map-pref-learning,liu-moi:c:mallows-mixtures},
or the Plackett--Luce model, which is similar in spirit to our
IAMs~\citep{zha-pie-xia:c:plackett-luce-icml16,liu-zha-lia-lu-xia:plackett-luce-aaai19,zha-xia:c:plackett-luce-nips19,ngu-zha:c:plackett-luce-aaai23}.
There is also literature on learning voting rules, but it is quite
distant from our work and we only mention a single paper on this topic~\citep{car-feh:c:learning-approval-rules}.

\section{Preliminaries}\label{sec:prelim}

For a positive integer $t$, by $[t]$ we mean the set
$\{1, \ldots, t\}$.  Given two numbers $a$ and $b$, we write $[a;b]$
to denote a closed interval between $a$ and $b$. For some
probabilistic event $X$, we write $\prob(X)$ to denote the probability
that it occurs.  Given a random variable $X$ and some probability
distribution $D$, we write $X \sim D$ to indicate that~$X$ is
distributed according to~$D$.  In particular, we use the following
standard distributions:
\begin{enumerate}
\item $U(a,b)$ is the uniform distribution over interval $[a;b]$.
\item Under $\mathit{Bernoulli}(p)$ we draw $1$ with probability $p$ and $0$ with probability
  $1-p$.
\item Under the categorical distribution
  $\mathit{Cat}(p_1,\ldots, p_k)$, for each $i \in [k]$ the
  probability of drawing $i$ is $p_i$ (hence we require all $p_i$
  values to be nonnegative and to sum up to $1$).
\end{enumerate}

\paragraph{Elections}
An \emph{(approval) election} is a pair $E=(C, V)$, where
$C=\{c_1, c_2, \ldots, c_m\}$ is a set of \emph{candidates} and
$V = (v_1, \ldots, v_n)$ is a collection of voters.  Each voter $v_i$
has a \emph{vote} $A(v_i) \subseteq C$ (also called an \emph{approval
  ballot}), i.e., a set of candidates that this voter approves.  By
$v_i(c_j)$ we mean the 1/0 value indicating whether $c_j$ is included
in $A(v_i)$ or not.
For each candidate $c \in C$, we write $V(c)$ to denote the set of
voters that approve $c$. The value $|V(c)|$ is known as the approval
score of $c$; $\nicefrac{|V(c)|}{n}$ is the probability that a
random voter approves $c$.
Given two votes, $X, Y \subseteq C$, their Hamming distance is
$\ham(X,Y) = |X \setminus Y| + |Y \setminus X|$, i.e., it is the
number of candidates approved in exactly one of them.

\paragraph{Probabilistic Models of Elections}
For a set of candidates $C$, we write $\calD(C)$ to denote the family
of probability distributions over the subsets of $C$, i.e., over the
votes with candidates from $C$.  For a distribution $D \in \calD(C)$
and a vote $X \subseteq C$, $\prob(X \cond D)$ is the probability of
generating vote $X$ under $D$.
For an election $E = (C,V)$, where $V = (v_1, \ldots, v_n)$,
$\prob(E \cond D)$ is the probability of generating $E$, provided that
each of its votes is drawn from $D$ independently:
\begin{equation}
  \textstyle \prob(E \cond D) = \prod_{i \in [n]} \prob(A(v_i) \cond D).\label{eq:prob}
\end{equation}
$\prob(E \cond D)$ is %
called the \emph{likelihood} of generating $E$ under $D$.  We will
also be interested in $\ln( \prob(E \cond D) )$, i.e., the
log-likelihood of generating $E$.

\begin{remark}
  We view the voters as non-anonymous. To see what this entails,
  consider elections $E' = (C,V')$ and $E'' = (C,V'')$, with
  $C = \{a,b\}$, $V' = (v'_1, v'_2)$, and $V'' = (v''_1, v''_2)$,
  where:
  \begin{align*}
    A(v'_1) = \{a,b\},\ A(v'_2) = \{b\}, && \text{and}&&
    A(v''_1) = \{b\},\ A(v''_2) = \{a,b\}.                      
  \end{align*}
  In our model, these two elections are distinct, but they would be
  equal if one viewed the voters as anonymous (it would only matter
  how many particular votes were cast, and not who cast them).
  The choice of the voter model does not affect our results: The problems of maximizing
  the probability of generating a given election under both models are
  equivalent (the respective probabilities only differ by a product of
  some binomial coefficents that only depend on the election's votes).

\end{remark}

For a candidate set~$C = \{c_1, \ldots, c_m\}$ and a number of
voters~$n$, a \emph{statistical culture} is a probability distribution
over elections with this candidate set and $n$ voters. By a small
abuse of notation, we will also refer to the distributions from
$\calD(C)$ as statistical cultures. Indeed, given $D \in \calD(C)$ we
can always sample an election by drawing the votes from $D$
independently.
The
following cultures from $\calD(C)$ are particularly relevant:
\begin{description}
\item[$\boldsymbol{p}$-Impartial Culture ($\boldsymbol{p}$-IC).] Under
  $p$-IC, for each voter $v$ and candidate $c$ we have that
  $v(c) \sim \mathit{Bernoulli}(p)$, i.e., each voter approves each
  candidate with probability $p$.

\item[$\boldsymbol{\phi}$-Hamming.] This distribution is parameterized
  by a central vote $U \subseteq C$ and a parameter $\phi \in [0;1]$.
  The probability of generating voter $v$'s vote is proportional to
  $\phi^{\ham(U,A(v))}$. $\phi$-Hamming is sometimes referred to as the
  $\phi$-noise
  model~\citep{szu-fal-jan-lac-sli-sor-tal:c:sampling-approval-elections}.
\item[$\boldsymbol{(p,\phi)}$-Resampling.] The resampling model is
  parameterized by a central vote $U \subseteq C$, resampling
  probability $\phi \in [0;1]$, and approval probability
  $p \in [0;1]$.  To generate a voter's $v$, vote $A(v) \subseteq C$,
  we do as follows: First, we let $A(v)$ be equal to $U$. Then,
  independently for each $c \in C$, with probability $\phi$ we replace
  value $v(c)$ with one sampled from $\mathit{Bernoulli}(p)$.

\end{description}
Impartial culture and the Hamming model are part of the folk
knowledge, although the Hamming model was recently studied by
\citet{car-kak-kar-kri:j:noise-approval} and
\citet{szu-fal-jan-lac-sli-sor-tal:c:sampling-approval-elections}. The
resampling model is due to
\citet{szu-fal-jan-lac-sli-sor-tal:c:sampling-approval-elections}.
The Hamming model is analogous to the classic Mallows model from the
world of ordinal elections~\citep{mal:j:mallows}, albeit
\citet{szu-fal-jan-lac-sli-sor-tal:c:sampling-approval-elections}
advocate using the resampling model instead.

\paragraph{Mixture Models for Elections}
Let $C$ be some set of candidates.  Given a family of $K$
distributions $D_1, \ldots, D_K \in \calD(C)$ and probabilities
$p_1, \ldots, p_K$ (where $\sum_{k=1}^K p_k = 1$), we can form
the following  \emph{mixture model}:
(1) We draw a number $k \sim \mathit{Cat}(p_1, \ldots, p_K)$ and, then,
(2) we draw a vote $X \sim D_k$.
We call $D_1, \ldots, D_K$ the components of this model, and $K$ is their number.

\section{Independent Approval Model}\label{sec:iam}

In this section we present the independent approval model and argue
that it generalizes all the statistical cultures from
\Cref{sec:prelim}.  We stress that it was already studied, e.g., by
\citet{lac-mal:j:approval-shortlisting} and
\citet{xia:t:linear-multiwinner}. Let $C = \{c_1,\ldots, c_m\}$ be the
candidate set:
\begin{description}

\item[$\boldsymbol{(p}_{\boldsymbol{1}}\boldsymbol{,\ldots, }
  \boldsymbol{p}_{\boldsymbol{t}}\boldsymbol{)}$-Independent Approval
  Model.]  In this model, abbreviated as $(p_1,\ldots, p_t)$-IAM, the
  candidate set is partitioned into $t$ disjoint groups,
  $C_1, \ldots, C_t$, and for each group $C_j$, each candidate
  $c_i \in C_j$ is approved independently, with probability $p_j$. We
  use the name \emph{$t$-parameter IAM} %
  when we disregard specific probability values.
\end{description}

The $(p_1, \ldots, p_m)$-IAM, where every candidate has his or her
individual approval probability, is a particularly natural special
case of the independent approval model, which we refer to as
\emph{full IAM}. Further, $p$-IC and the $(p,\phi)$-resampling models
also are special cases of IAM. For the former, simply take a single
candidate group with approval probability $p$. For the latter, note
that $(p,\phi)$-resampling with central vote $U$ is equivalent to the
$(p_1,p_2)$-IAM with:
\begin{align*}
  p_1 = (1-\phi) + \phi \cdot p, &\quad \text{ and }\quad p_2 = \phi \cdot p,
\end{align*}
where $C_1 = U$ and $C_2 = C \setminus U$. 
In the other direction,
$(p_1,p_2)$-IAM with candidate groups $C_1$ and $C_2$, where
$p_1 \geq p_2$, is equivalent to $(p,\phi)$-resampling with
$\phi = 1 - (p_1-p_2)$,
$p = \nicefrac{p_2}{1-(p_1-p_2)}$,
and central vote $U = C_1$ (note that as $p_1 \geq p_2$, we have that
$p$ and $\phi$ are guaranteed to be between $0$ and $1$).
Consequently, resampling and $2$-parameter IAM are equivalent.

The $\phi$-Hamming model is a special case of $2$-parameter
IAM. Putting it in our language,
\citet{car-kak-kar-kri:j:noise-approval} have shown that for an
approval probability $p \in [0.5;1]$, $(p,1-p)$-IAM---with candidate
groups $C_1$ and $C_2$---is equivalent to $\phi$-Hamming with
$\phi = \frac{1-p}{p}$ and central vote $U = C_1$.  Similarly,
$\phi$-Hamming with central vote $U$ is equivalent to
$(p_1,p_2)$-resampling with candidate groups $C_1 = U$ and
$C_2 = C \setminus U$, $p_1 = \frac{1}{1+\phi}$ and
$p_2 = 1-p_1 = \frac{\phi}{1+\phi}$.

Altogether, we have the following hierarchy of expressivity of the IAM
models (we view our models as sets of distributions,
for all possible choices of parameters; e.g., by
$p$-IC we mean all the impartial culture distributions for all
approval probabilities $p$):
\begin{align*}
  p\hbox{-}\text{IC} \subset \phi\hbox{-}\text{Hamming} &\subset
                       (p,\phi)\hbox{-}\text{Resampling} = (p_1,p_2)\hbox{-}\text{IAM} \\
  &\subset (p_1,p_2,p_3)\hbox{-}\text{IAM} \subset \cdots \subset (p_1,\ldots, p_m)\hbox{-}\text{IAM}.
\end{align*}

\begin{remark}
  Later on %
  we consider mixture models based on IAM variants. For
  example, by %
  $2$-full-IAM we mean %
  a mixture model
  with two full-IAM components. %
  This
  should not be confused with $2$-parameter IAM, which %
  is a   single-component resampling model.
\end{remark}

\section{Learning Algorithms}

Let us now focus on the following task: We are given an election
$E = (C,V)$ with candidate set $C = \{c_1, \ldots, c_m\}$, and voter
collection $V = (v_1, \ldots, v_n)$. We also have a family of
distributions from $\calD(C)$. Our goal is to find a distribution from
this family that maximizes the probability of generating election
$E$. We will first solve this problem for each of the special cases of
IAM from the previous two sections, and then we will consider IAM
mixture models.

\subsection{Learning a Single IAM Model}\label{sec:learn-single-iam}

Let $E = (C,V)$ be an election, as described above. For each set of
candidates $B \subseteq C$, let $\app_E(B) = \sum_{c \in B}|V(c)|$ be
the total number of approvals that members of $B$ receive, and let
$\pr_E(B) = \frac{\app(B)}{n|B|}$ be the probability that a random
voter approves a random candidate from~$B$. By $E(B)$, we mean
election $E$ restricted to the candidate set~$B$.  Consider a
partition of $C$ into sets $X$ and $Y$, and let $D_X \in \calD(X)$ and
$D_Y \in \calD(Y)$ be two IAMs with $t_1$ and $t_2$ parameters,
respectively. Further, let $D_{XY} \in \calD(C)$ be the
$(t_1+t_2)$-parameter IAM that generates approvals for candidates from
$X$ according to $D_X$ and for those from $Y$ according to $D_Y$. We
have:
\begin{equation}
  \prob(E \cond D_{XY}) = \prob(E(X) \cond D_X) \cdot \prob(E(Y) \cond
  D_Y).
  \label{eq:iam-decompose}
\end{equation}
For each $t \in [|C|]$ and each partition of $C$ into
$C_1, \ldots, C_t$, we write IAM$(C_1, \ldots, C_t)$ to refer to the
$t$-parameter IAM that uses this partition and for each $i \in [t]$,
the probability of approving a candidate from $C_i$ is
$p_i = \pr_E(C_i)$. As per \Cref{eq:iam-decompose}, we have:
\begin{equation*}
  \textstyle
  \prob(E \cond  \text{IAM}(C_1, \ldots, C_t)) = %
                 \textstyle \prod_{i \in [t]} \prob(E(C_i) \cond \pr_{E(C_i)}(C_i)\hbox{-}\text{IC}).
\end{equation*}
Intuitively, if we want a $t$-parameter IAM that maximizes the
probability of generating a given election, it suffices to use
IAM$(C_1, \ldots, C_t)$ for an appropriate partition of the
candidate set.
Next we show this fact formally and argue how to find optimal
partitions (for the $\phi$-Hamming model we use a different approach).

\subsubsection{Impartial Culture and the Full IAM Model}
For impartial culture, finding the parameter that maximizes the
probability of generating a given election $E$ is easy: It suffices to
use $p$-IC with $p$ equal to the proportion of approvals in the
election
(this is a standard fact from
statistics, often expressed in different contexts).

\begin{proposition}\label{prop:learn:p-ic}
  Let $E = (C,V)$ be an approval election.
  Probability $\prob(E \cond q\hbox{-}\text{IC})$
  is maximized for $p = \pr_E(C)$.
\end{proposition}
\begin{proof}
  Consider election $E = (C,V)$ and let $p = \pr_E(C)$.  Let $m$ be
  the number of candidates and let $n$ be the number of voters.  For a
  number $q \in [0;1]$, the probability of generating $E$ under the
  $q$-IC model is as follows:
  \begin{align*}
    f(q) =  q^{\app_E(C)}(1-q)^{nm-\app_E(C)} 
         =  q^{pnm}(1-q)^{(1-p)nm} 
         =  \left(q^p(1-q)^{1-p}\right)^{nm}.
  \end{align*}
  In other words, there are $\app_E(C) = pnm$ approvals
  in the election and each of them is correctly generated with
  probability $q$. Similarly, each of the
  $nm - \app_E(C) = (1-p)nm$ nonapprovals is generated
  correctly with probability $1-q$. As, from our point of view, $m$
  and $n$ are constants, $f(q)$ is maximized for the same argument as
  $h(q) = q^p(1-q)^{1-p}$. The derivative of $h(q)$ is:
  \begin{align*}
    h'(q) = pq^{p-1}(1-q)^{1-p} - q^p(1-p)(1-q)^{-p} 
          =q^p(1-q)^{-p} \cdot \left( pq^{-1}(1-q) - (1-p) \right),
  \end{align*}
  and it assumes value $0$ when:\footnote{We disregard the cases where
    either all the voters approve all the candidates or neither of the
    voters approves any of the candidates, for which, respectively,
    $q = 1$ and $q = 0$ are immediately seen to be the values
    maximizing $\prob(E \cond q\hbox{-}\text{IC})$.}
  \begin{align*}
    pq^{-1}(1-q) = (1-p).
  \end{align*}
  This holds when $p(1-q) = q(1-p)$, which is equivalent to $q = p$.
  Hence $\prob(E \cond q\hbox{-}\text{IC})$ is maximized for $q = p = \pr_E(C)$.
\end{proof}

Consequently, to learn the parameters of an IAM for a given election,
it suffices to find a partition of the candidates.

\begin{proposition}\label{pro:iam-decompose}
  Let $E = (C,V)$ be an election. %
  For each
  $t \in [|C|]$ there is a partition of $C$ into $C_1, \ldots, C_t$
  such that IAM$(C_1, \ldots, C_t)$ maximizes the probability of
  generating $E$ among $t$-parameter IAMs.
\end{proposition}
\begin{proof}
  Let us fix a value of $t \in [|C|]$. Consider a $t$-parameter IAM
  model that maximizes the probability of generating $E$. By
  definition, the model is parameterized by a partition of $C$ into
  subsets $C_1, \ldots, C_t$, and probabilities $p_1, \ldots, p_t$,
  such that for each $i \in [t]$, each candidate from $C_i$ is
  approved with probability $p_i$. By \Cref{eq:iam-decompose}, each
  $p_i$ maximizes the probability of generating $E(C_i)$ under the
  impartial culture model. Thus, by \Cref{prop:learn:p-ic}, for each
  $i \in [t]$ we have $p_i = \pr_{E(C_i)}(C_i)$. So $D$ is
  IAM$(C_1, \ldots, C_t)$.
\end{proof}

Consequently, the full IAM that maximizes the probability of
generating a given election uses parameters where each candidate is
approved with probability equal to the fraction of its approvals.

\begin{corollary}\label{cor:learn:iam}
  Let $E = (C,V)$ be an election, where %
  $C = \{c_1, \ldots, c_m\}$ and $n$ voters.  Probability
  $\prob(E \cond (p_1,\ldots, p_m)\hbox{-}\text{IAM})$ is maximized if
  for each $i \in [m]$ we have $p_i = \nicefrac{|V(c)|}{n}$.
\end{corollary}

\subsubsection{Hamming Model}
For each $\phi \in [0;1]$ and each vote $U$, we let $\phi$-$\Ham(U)$
denote the $\phi$-Hamming model with central vote~$U$.  The
probability of generating vote $A(v)$ under $\phi$-$\Ham(U)$ is:
\[
  \textstyle
  \prob\big(A(v) \cond \phi\hbox{-}\Ham(U)\big) = \frac{1}{(1+\phi)^m} \phi^{\ham\big(U,A(v)\big)}
\]
(the normalizing constant is derived, e.g., by
\citet{car-kak-kar-kri:j:noise-approval}), and the probability of
generating election~$E$~is:
\begin{equation}
  \textstyle
  f_u(\phi) = \prob\big(E \cond \phi\hbox{-}\Ham(U)\big) = \frac{1}{(1+\phi)^{mn}} \phi^{\sum_{i=1}^n{\ham\big(U,A(v_i)\big)}}.
\end{equation}
For each fixed $\phi$, this value is maximized when the exponent,
$\sum_{i=1}^n{\ham\big(U,A(v_i)\big)}$, is minimized. This happens for
central vote $U$ such that for each candidate $c_i$, $c_i$ belongs to
$U$ if and only if at least half of the voters approve $c_i$. We refer
to such a central vote as \emph{majoritarian}.
Let us fix $U$ to be
majoritarian and let $h = \sum_{i=1}^n{\ham\big(U,A(v_i)\big)}$. By definition,
we have $h \leq \nicefrac{mn}{2}$, and we assume that $h > 0$
(otherwise it suffices to take $\phi = 0$ to maximize the probability
of generating $E$).  The derivative of $f_u(\phi)$ (with respect to
$\phi$) is:
\begin{align*} \textstyle
  f'_u(\phi)  = \textstyle\frac{h\phi^{h-1}(1+\phi)^{mn} - mn(1+\phi)^{mn-1}\phi^h}{(1+\phi)^{2mn}} 
  \textstyle = \frac{\phi^{h-1}(1+\phi)^{mn-1}}{(1+\phi)^{2mn}}(h(1+\phi) - mn\phi).
\end{align*}
By analyzing the final term, one can verify that its value is~$0$
exactly for $\phi = \frac{h}{mn-h}$, for smaller $\phi$ it is
positive, and for larger $\phi$ it is negative.  Hence, we have
the following result.

\begin{proposition}
  Let $E = (C,V)$ be an approval election with $m$ candidates and $n$
  voters, where $V = (v_1, \ldots, v_n)$. Let $U$ be the majoritarian
  central vote for $E$ and let $h = \sum_{i=1}^n{\ham(u,v_i)}$. The
  probability of generating $E$ using $\phi$-Hamming model is
  maximized for the majoritarian central vote and
  $\phi = \nicefrac{h}{mn-h}$.
\end{proposition}

\subsubsection{Resampling and Other IAMs}
Next, let us consider learning the parameters of $2$-parameter IAMs
(or, equivalently, of the resampling models). The solution is
intuitive, but requires a more careful proof. In particular, the next
theorem restricts the parameter space that we need to analyze.

\begin{theorem}\label{thm:2-iam}
  Let $E = (C,V)$ be an election with candidate set
  $C = \{c_1, \ldots, c_m\}$, and
  voter collection $V = (v_1, \ldots, v_n)$, such that
  $|V(c_1)| \geq |V(c_2)| \geq \cdots \geq |V(c_m)|$. Then
  $\prob(E \cond (p_1,p_2)\hbox{-}\text{IAM})$ is maximized for some
  $m' \in [m]$ and:
\begin{align*}
  p_1 &= \nicefrac{|V(c_1)|+ \cdots + |V(c_{m'})|}{nm'}, \quad C_1 = \{c_1, \ldots, c_{m'}\}, \\ 
  p_2 &= \nicefrac{|V(c_{m'+1})|+ \cdots + |V(c_{m})|}{n(m-m')}, \quad  C_2 = C \setminus C_1.
\end{align*}
\end{theorem}
 In other words, this theorem says that a 2-parameter IAM that best
 captures a given election puts the most and the least approved
 candidates in two different groups.

\begin{proof}[Proof of \Cref{thm:2-iam}]
  Let the notation be as in the statement of the theorem. For a subset
  $B \subseteq C$ of candidates, let $\app(B) = \sum_{c \in B}|V(c)|$
  be the total number of approvals that the candidates from $B$
  receive within election $E$, and let
  $p(B) = \nicefrac{\app(B)}{|C|n}$ be the probability that a randomly
  selected voter approves a member of $B$. We let $t = \app(C)$ be the
  overall number of approvals cast within~$E$.

  Let us fix some partition of $C$ into $C'_1$ and
  $C'_2 = C \setminus C'_1$, where neither $C'_1$ nor $C'_2$ is empty,
  and where $p(C'_1) \geq p(C'_2)$.
  The probability that $E$ is generated under IAM$(C'_1,C'_2)$ is as
  follows (where $x = \app(C'_1)$, $A_1 = |C'_1|n$, and
  $A_2 = |C'_2|n$; note that $t-x = \app(C'_2)$):
  \begin{align*}
    \textstyle
    f(x)\!=\!  {\left(\frac{x}{A_1}\right)\!}^x\!{\left(\frac{A_1-x}{A_1}\right)\!}^{A_1-x}\!
    {\left(\frac{t-x}{A_2}\right)\!}^{t-x}\!\left(\frac{A_2-(t-x)}{A_2}\right)^{A_2-(t-x)}.
  \end{align*}
  Intuitively, $A_1$ and $A_2$ are the numbers of approval-disapproval
  decisions that voters have to make regarding the candidates in $C_1$
  and $C_2$, respectively.  The first two factors of $f(x)$ give the
  probability that the decisions regarding candidates in $C_1$ are
  made as in $E$ (the former corresponds to approvals and the latter
  to disapprovals), and the next two factors give analogous
  probability for the candidates in $C_2$.  We consider $f$ as a
  function defined for real arguments $x$ such that
  $\frac{x}{A_1} \geq \frac{t-x}{A_2}$ and $x \leq A_1$ (the
  assumption that $\frac{x}{A_1} \geq \frac{t-x}{A_2}$ corresponds to
  $p(C'_1) \geq p(C'_2)$ and can be equivalently expressed as
  $x \geq \frac{A_1t}{A_1+A_2} = \frac{A_1t}{nm}$). Later on, we will
  prove the following claim.
  \begin{claim}\label{claim:monotonicity}
    Function $f$, defined on arguments $x$ such that
    $\frac{x}{A_1} \geq \frac{t-x}{A_2}$, is nondecreasing.
  \end{claim}

  Let us consider some two candidates $c_1 \in C'_1$ and
  $c_2 \in C'_2$, and a partition obtained from $C'_1, C'_2$ by
  swapping their membership in these sets:
  \begin{align*}
    C''_1 = (C'_1 \setminus \{c_1\}) \cup \{c_2\},
    \text{ and }
    C''_2 = (C'_2 \setminus \{c_2\}) \cup \{c_1\}.
  \end{align*}
  We observe that the probability of generating $E$ under
  $\text{IAM}(C''_1,C''_2)$ is equal to $f(x-\app(c_1)+\app(c_2))$.
  This means that if $\app(c_2) \geq \app(c_1)$, then swapping $c_1$
  and $c_2$ does not decrease the probability of generating
  $E$. Consequently, if we take partition $C^\star_1$ and $C^\star_2$
  of $C$ such that $C^\star_1$ contains $|C'_1|$ candidates with the
  highest approval scores and $C^\star_2$ contains the remaining ones,
  then $\text{IAM}(C^\star_1,C^\star_2)$ maximizes the probability of
  generating $E$ among $2$-parameter IAMs that partition the
  candidates into groups with $|C'_1|$ and $|C'_2|$ candidates, where
  the former group has at least as high approval probability as the
  latter one (indeed $C_1^\star, C_2^\star$ can be obtained from any
  such parition by a sequence of swaps that do not decrease the
  probability of generating $E$). This gives exactly the statement of
  our theorem.

  It remains to show that \Cref{claim:monotonicity} holds. To this
  end, we consider $g(x) = \ln( f(x) )$, where $\ln(\cdot)$ is the natural logarithm. We have:
  \begin{align*}
    \textstyle
    g(x)  &= \textstyle  x\ln\left( \frac{x}{A_1} \right) + (A_1-x)\ln\left( \frac{A_1-x}{A_1} \right) \\
          &+ \textstyle (t-x) \ln \left( \frac{t-x}{A_2} \right) + (A_2-(t-x))\ln\left(\frac{A_2-(t-x)}{A_2}\right).
  \end{align*}
  Naturally, $g(x)$ is nondecreasing if and only if $f$ is. Next, we
  compute the derivative of $g$ (see also the explanations below):
  \begin{align*}
    \textstyle
    g'(x) &= \textstyle  \ln\left( \frac{x}{A_1} \right) -\ln\left( \frac{A_1-x}{A_1} \right) \\
          &  \textstyle -\ln\left( \frac{t-x}{A_2} \right) + \ln\left(\frac{A_2-(t-x)}{A_2}\right) \\
          &= \textstyle \ln\left( \frac{x}{A_1-x} \cdot \frac{A_2-(t-x)}{t-x} \right) \geq 0.
  \end{align*}
  For the final inequality, note that we have assumed that
  $\frac{x}{A_1} \geq \frac{t-x}{A_2}$.  This is equivalent to
  $\frac{x}{(A_1-x)+x} \geq \frac{t-x}{A_2 - (t-x) + (t-x)}$, which
  itself implies
  $\frac{(A_1-x)+x}{x} \leq \frac{A_2 - (t-x) + (t-x)}{t-x}$ and,
  after simplification,
  $\frac{A_1-x}{x} \leq \frac{A_2 - (t-x)}{t-x}$. Hence, the argument
  under the logarithm in the above inequality is at
  least~$1$. Finally, since $g'(x)$ is nonnegative, $g$ is
  nondecreasing and so is~$f$.  
\end{proof}

Intuitively, \Cref{thm:2-iam} says that if we want to find the
parameters of a $2$-parameter IAM that maximize the probability of
generating a given election $E = (C,V)$, then there are only
polynomially many options to try. Namely, using the notation from
\Cref{thm:2-iam}, it suffices to try all $O(m)$ choices of $m'$, each
giving a different candidate partition, and---among those---select the
one that leads to maximizing the probability of generating $E$.

\begin{corollary}
  There is a polynomial-time algorithm that given an election
  $E = (C,V)$ finds the parameters of a $2$-parameter IAM that
  maximizes the probability of generating~$E$.
\end{corollary}

Using the ideas from the proof of \Cref{thm:2-iam}, we also obtain an
analogous result for IAMs with arbitrary number~$t$ of parameters.

\begin{theorem}\label{thm:t-iam}
  Let $E = (C,V)$ be an election with $m$ candidates and $n$
  voters. For each $t \in [m]$,
  $\prob(E \cond (p_1,p_2, \ldots, p_t)\hbox{-}\text{IAM})$ is
   maximized for $p_1, \ldots, p_t$ and a partition of $C$ into
   $C_1, \ldots, C_t$ such that
  for each $i \in [t]$ we have
  $p_i = \nicefrac{\sum_{c \in C_i}|V(c)|}{n|C_i|}$ and for each
  $i \in [t-1]$ and every two candidates $a \in C_i$ and
  $b \in C_{i+1}$ we have $|V(a)| \geq |V(b)|$.
\end{theorem}
\begin{proof}
  Consider an election $E = (C,V)$ and some arbitrary partition of $C$
  into sets $C_1, \ldots, C_t$.  Assume that there are two candidates,
  $c_1 \in C_1$ and $c_x \in C_i$, where $i \in [t] \setminus \{1\}$,
  such that $c_x$ is approved by more voters than $c_1$. Form a
  partition $C'_1, \ldots, C'_t$ that is identical to
  $C_1, \ldots, C_t$, except that $c_1$ is in $C_i$ and $c_x$ is in
  $C_1$. Then, by the same argument as in the proof of
  \Cref{thm:2-iam}, we have the following (note that
  $C'_1 \cup C'_i = C_1 \cup C_i$):
  \begin{align*}
    \prob(E(C'_1 \cup C'_i)  \cond \text{IAM}(C'_1, C'_i) 
                            \geq \prob(E(C_1 \cup C_i)  \cond \text{IAM}(C_1, C_i)).                              
  \end{align*}
  Thus, by applying \Cref{eq:iam-decompose}, we see that the
  probability that IAM$(C'_1, \ldots, C'_t)$ generates $E$ is at
  least as high as that for IAM$(C_1, \ldots, C_t)$. By applying this
  reasoning repeatedly, we can transform $C_1, \ldots, C_t$ into a
  partition that satisfies the conditions from the theorem statement,
  without ever decreasing the probability of generating election $E$.
  This completes the proof as the initial choice of $C_1, \ldots, C_t$
  was arbitrary.
\end{proof}

Based on \Cref{thm:t-iam}, we derive a dynamic-programming algorithm
that computes a $t$-parameter
IAM that maximizes the probability of generating a given election.

\begin{theorem}\label{thm:t-iam-algo}
  There is a polynomial-time algorithm that given an election
  $E = (C,V)$ and an integer $t \in [|C|]$ finds the parameters of a
  $t$-parameter IAM that maximizes the probability of
  generating~$E$.
\end{theorem}
\begin{proof}
  Let $E = (C,V)$ be our input election, where
  $C = \{c_1, \ldots, c_m\}$, and $V = (v_1, \ldots, v_n)$, and let
  $t$ be the number of IAM parameters that we are to optimize. Without
  loss of generality, we assume that
  $|V(c_1)| \geq |V(c_2)| \geq \cdots \geq |V(c_m)|$.  For each
  $i, j \in [m]$, $i \leq j$, by $C[i,j]$ we mean the set
  $\{c_i, \ldots, c_j\}$.  For each subset $B \subseteq C$ of
  candidates and each integer $\ell \in [t]$, we let $f(B,\ell)$ be
  the highest possible probability of generating $E(B)$ using an
  $\ell$-parameter IAM. We will show an algorithm for computing
  $f(C,t)$.

  By \Cref{prop:learn:p-ic}, we know that for each subset
  $B \subseteq C$ of candidates, we can compute $f(B,1)$ in polynomial
  time (it corresponds to using an impartial culture model with
  approval probability $\pr_{E(B)}(B)$). Next, for each
  $\ell \in [t] \setminus \{1\}$ and $j \in [m]$, $j \geq \ell$, we
  have the following recursive relation between the values of $f$ (it
  is true due to \Cref{thm:t-iam}):
  \[
    f(C[1,j],\ell) = \max_{i \colon \ell-1 \leq i < j} \bigg( f(C[1,i],\ell-1)\cdot f(C[i+1,j],1) \bigg).
  \]
  Using this equation and standard dynamic-programming techniques, we
  can compute in polynomial time both the value $f(C,t)$ and the
  partition of $C$ into $C_1, \ldots, C_t$ such that:
  \[
    f(C,t) = f(C_1,1) \cdot f(C_2,1) \cdots f(C_t,1).
  \]
  IAM$(C_1, \ldots, C_t)$ is the model that maximizes the probability
  of generating $E$ among the $t$-parameter IAMs.
\end{proof}

We should remark that while we have expressed the algorithm in the
above theorem in terms of multiplication of the values of $f$, in
practice it is better to optimize the logarithms of the values of $f$
and, hence, use addition. This means optimizing log-likelihood of
generating $E$ instead of the actual likelihood. The two approaches
are formally equivalent, but the former is more stable numerically.

\subsection{Mixture Models: Expectation Maximization}\label{sec:em}

In this and the next section we consider two different approaches of
learning mixtures of IAMs. Here we start with
one of the most standard and widely adopted machine learning
algorithms used for estimating the values of parameters in statistical
models, the Expectation-Maximization (EM) algorithm
\citep{dempter1977maximum}. The algorithm is useful when there are
some missing or unobserved latent variables in the model. In the
context of learning mixtures of distributions, these latent variables
correspond to the assignment of the data points to the mixture
components that generated them. The algorithm first guesses the
parameters of the model and then iteratively improves the estimation
trying to maximize the log-likelihood of generating the observed
data. Each iteration consists of two steps:
\begin{description}
\item[E-step (Expectation),] where the algorithm computes the
  posterior probability (soft assignment) that each data point was
  generated by each mixture component.
\item[M-step (Maximization),] where the algorithm updates the
  parameters of the mixture model to maximize the expected
  log-likelihood given the probabilities from the E-step.
\end{description}
Intuitively, after each iteration the likelihood of generating the
observed data under the current model parameters increases. The
algorithm alternates between the E- and M-steps until convergence
(i.e.,\ until two consecutive iterations return the same model
parameters up to some negligible $\epsilon$), obtaining a local
maximum of the log-likelihood.

Let us now consider how the EM algorithm can be used for learning
mixtures of IAM models. Let us fix some election $E$ and the number
$K$ of components. For $k \in [K]$, the $k$-th component has
parameters $\theta_k = (\alpha_k, (p_1, p_2, \ldots, p_m))$,
where %
$\alpha_k$ is the weight of the component
($\sum_{k \in [K]} \alpha_k = 1$) and $p_1, \ldots, p_m\in [0;1]$ are
the probabilities of approving candidates $c_1, \ldots, c_m$ (these
probabilities are independent only for the full IAM model; however,
potential dependencies between them are not relevant for the
E-step). Given a $k$-th mixture component and its parameters
$\theta_k$, the probability of generating vote $A(v)$ by this
component is:
\begin{equation*}
  \mathbb{P}(A(v) | \theta_k) =
  \alpha_k \cdot \left(\Pi_{j\in [m]: c_j\in A(v)} p_{j} \right) \left(\Pi_{j\in [m]: c_j\notin A(v)} (1-p_j) \right).
\end{equation*}
Let us denote the posterior probability of generating votes from $E$, computed
in the E-step of the algorithm---further called the \emph{soft
  assignment} of votes from $E$---as
$\gamma = \{\gamma_{v, k}\}_{v\in V, k\in [K]}$. For each voter $v$
and the $k$-th component, $\gamma_{v, k}$ is equal to the conditional probability of
generating $v$ by the $k$-th component subject to the fact that $v$ has
been generated by \emph{some} component. Formally:
\begin{equation*}
  \textstyle
  \gamma_{v,k} = \frac{\mathbb{P}(A(v) | \theta_k)}{\sum_{j} \mathbb{P}(A(v) | \theta_j)}.
\end{equation*}
For clarity of the notation, for each $k \in [K]$ let us denote the total weight
of the votes assigned to the $k$-th component
as $\gamma_k = \sum_{v\in V} \gamma_{v, k}$. Note that
$\sum_{k \in [K]} \gamma_k = n$.

For the M-step, we need to find the parameters
$\theta_1, \theta_2, \ldots, \theta_K$ maximizing the expected
complete log-likelihood up to the computed soft assignment, given by
the following formula:
\begin{equation*}
  \textstyle
    \sum_{v\in V} \sum_{k \in [K]} \gamma_{v,k} \ln(\mathbb{P}(A(v) | \theta_k)).
\end{equation*}
Note that in the above formula the gamma variables should be viewed as
constants; their values depend on theta values computed in the
previous iteration of the algorithm, not on the ones we currently
search for.

Let us first focus on the $\alpha_k$ parameter of each component
$k \in [K]$. We have the following result here:

\begin{proposition}\label{prop:maximizing_log_likelihood_for_iam_mixtures}
  For each mixture of $K$ IAM distributions, the expected
  log-likelihood of generating the observed data is maximized if for
  each $k \in [K]$ it holds that $\alpha_k = \nicefrac{\gamma_k}{n}$. %
\end{proposition}

\begin{proof}
  We need to maximize the following expression:
  \begin{align*}
    \textstyle
    \sum_{v\in V} \sum_{k \in [K]} \gamma_{v,k} \ln(\mathbb{P}(A(v) &| \theta_k)) = \\
      \textstyle \sum_{v\in V} \sum_{k \in [K]} \gamma_{v,k} \ln(\alpha_k) &+
      \textstyle \sum_{v\in V} \sum_{k\in [K]} \gamma_{v,k}
      \ln(\frac{\mathbb{P}(A(v) | \theta_k)}{\alpha_k}).
  \end{align*}
  Note that the value of
  $\nicefrac{\mathbb{P}(A(v) | \theta_k)}{\alpha_k}$ does not depend
  on the value of $\alpha_k$ for each $k \in [K]$. Conseqeuntly, the
  second part of the sum is irrelevant for the maximization and can be
  skipped. Hence, we need to maximize:
  \begin{align*}
    \textstyle
    \sum_{v\in V} \sum_{k \in [K]} \gamma_{v,k} \ln(\alpha_k) = \sum_{k \in [K]}
    \ln(\alpha_k)\cdot \gamma_k.
  \end{align*}
  To maximize this expression subject to the condition
  $\sum_{k \in [K]} \alpha_k = 1$, we can use the Lagrangian multiplier
  method. The corresponding Lagrangian is as follows:
  \begin{align*}
    \textstyle
    \mathcal{L}((\alpha_1, \ldots, \alpha_K), \lambda) &= \sum_{k \in [K]}
    \ln(\alpha_k) \cdot \gamma_k + \lambda(1 - \sum_{k \in [K]} \alpha_k).
  \end{align*}
  By straightforward computations (computing partial derivatives with
  respect to each $\alpha_k$ and comparing them to $0$), we obtain
  $\alpha_k = \nicefrac{\gamma_k}{\lambda}$ for each $k \in [K]$ and,
  since $\sum_{k \in [K]} \alpha_k = 1$, we have that:
  \begin{equation*}
    \textstyle
    \lambda = \sum_{k \in [K]} \gamma_{k} = n.
  \end{equation*}
  Hence:
  \begin{equation*}
    \textstyle
    \alpha_k = \frac{\gamma_k}{\lambda} = \frac{\gamma_k}{n} \text{, \quad for
    each }k \in [K].
  \end{equation*}
  which completes the proof.
\end{proof}

Let us now focus on optimizing the remaining parameters of IAM
components. Note that since we know the optimal weights of the
components and the other parameters are independent between different
components, we can now optimize the parameters of each IAM component
separately.

Let us fix component number $k \in [K]$ and compute the probabilities
maximizing the expected log-likelihood of generating a corresponding
soft assignment $\gamma$ using the $k$-th component. Note that for
each $v\in V$ we have that $\gamma_{v, k}$ is a rational number, i.e.,
it is a quotient of two integers. Let $Q$ be the least common multiple
of all the denominators of these quotients.  We now consider an
election $E_{\gamma, k}$ obtained from the initial election $E$ and
the given soft assignment $\gamma$ by multiplying each voter $v$,
$Q\cdot \gamma_{v,k}$ times (so that for each two votes $v$ and $v'$,
the proportion between the numbers of their copies is equal to
$\nicefrac{\gamma_{v, k}}{\gamma_{v', k}}$), which we will further
call an election \emph{induced by} $E$, $\gamma$ and the $k$-th component. We
will show that maximizing the
log-likelihood of generating the part of $\gamma$ corresponding to the
considered $k$-th component is equivalent to maximizing the
probability of generating~$E_{\gamma, k}$.

\begin{proposition}\label{prop:maximization-simplification}
  Let $E = (C,V)$ be an approval election with $m$ candidates and $n$ voters and
  let $\gamma$ be the soft assignment of votes to some $K$ mixture components.
  Let $E_{\gamma, k}$ be the election induced by $E$, $\gamma$ and some $k$-th
  component for $k \in [K]$. Then finding the $k$-th mixture component
  parameters maximizing the log-likelihood of generating the assignment is
  equivalent to finding the parameters maximizing the probability of generating
  $E_{\gamma, k}$.
\end{proposition}
\begin{proof}
  The formula for the expected complete log-likelihood is as follows:
  \begin{equation*}
    \textstyle
    \sum_{v\in V} \sum_{k \in [K]} \gamma_{v,k} \ln(\mathbb{P}(A(v) | \theta_k)).
  \end{equation*}
  Since the parameters of the mixture components are independent, our
  goal is to maximize the following expression for each $k \in [K]$:
  \begin{equation*}
    \textstyle
    \sum_{v\in V} \gamma_{v,k} \ln(\mathbb{P}(A(v) | \theta_k)).
  \end{equation*}
  The above formula is equivalent to the following one:
  \begin{equation*}
    \textstyle
    \Pi_{v\in V} \mathbb{P}(A(v) | \theta_k)^{\gamma_{v,k}}.
  \end{equation*}
  Then, by raising the above formula to the power of $Q$ (which does
  not affect the parameter values maximizing the expression) we
  obtain:
  \begin{equation*}
    \textstyle
    \Pi_{v\in V} \mathbb{P}(A(v) | \theta_k)^{Q\cdot \gamma_{v,k}}.
  \end{equation*}
  which is the probability of generating $E_{\gamma, k}$ by
  parameters~$\theta_k$.
\end{proof}

We also obtain direct formulas for optimal
parameters %
for other models from the preceding section (see \Cref{apdx:em}).

%
%
%
%

%
%
%
%
%
%
%
%
%
%

%
%
%
%
%
%
%
%
%

%
%
%
%
%
%
%
%
%
%
%

%
%
%
%
%
%
%
%
%
%
%
%
%
%

%

%

%
%
%
%
%
  
%
%
%
%
  
%
%
%
%

%

%
%
%

%
%
%

%
%
%

%
%
%

%
%

%
%
%
%
%
%
%
%
%
%
%
%
%
%
%
%
%
%
%
%

%

\subsection{Mixture Models: Bayesian Learning}\label{sec:bayes}

Consider an election $E = (C,V)$, with candidates $C = \left\{c_1, \ldots, c_m\right\}$ and voters $V = \left(v_1, \ldots, v_n\right)$.
So far we focused on learning parameter values that maximize the probability of generating $E$. Once learned,
model's parameters in these settings are, essentially, fixed constants. Alternatively, we can view the parameters themselves
as random variables. We can then estimate a distribution over the parameters that is compatible with the election
$E$. Learning the distribution over model's parameters conditioned on the observed data is the core concept in Bayesian
statistics.

One simple way to formalize a Bayesian model for an approval election is to postulate a \emph{generative process} for
the votes, i.e., a sampling procedure that describes our \emph{prior assumptions} about the distribution over the votes.
For example, consider the full IAM model parametrized by the vector of approval probabilities: $\left( p_1, p_2, \ldots
p_m \right)$. The generative process in this case starts with sampling each approval probability $p_i$ from a prior
distribution over its possible values. In this work we do not assume any a priori knowledge about the approval
probabilities. The prior distribution for approval probabilities is therefore the uniform distribution over the $[0, 1]$
interval. After sampling the parameters, the generative process samples the votes conditioned on the parameter values. In
the IAM model this conditional distribution is simply the Bernoulli distribution parametrized by the approval
probability. Together, these two steps give the following process:
\begin{enumerate}
  \item For all $c \in C$, sample the approval probability,  $p_c \sim U(0, 1)$.
  \item For all $v_i \in V, c \in C$, sample the vote outcome, $v_i(c) \sim \mathit{Bernoulli}(p_c)$.
\end{enumerate}
The generative process fixes the prior over model's parameters $p \left( p_1, \ldots, p_m \right)$ and the
data likelihood $p \left( V \mid p_1, \ldots, p_m \right)$. The Bayes theorem then gives a principled formula
for the \emph{posterior distribution} over the model's parameters: $p \left( p_1, \ldots, p_m \mid V \right)$.
This distribution summarizes our knowledge about values of the parameters, once we observed a set of votes $V$.

The Bayesian framework provides a flexible way to specify more complex generative processes. In particular, we can easily
write a generative process for a mixture of $K$ full IAM components (see \Cref{apdx:bayes_models} for other models):
\begin{enumerate}
  \item Sample component probabilities, $\left( \alpha_1, \ldots, \alpha_K \right) \sim \mathit{Dirich}(\mathbf{1}^K)$.
  \item For all $c \in C,\ k \in [K]$, sample the $k$-th component's
    approval probability for the candidate $c$, $p_{c, k} \sim U(0, 1)$.
  \item For all $v_i \in V$:
    \begin{enumerate}
    \item Sample the component index, $z \sim \mathit{Cat}(\alpha_1, \ldots, \alpha_K)$.
    \item For all  $c \in C$, sample $v_i(c) \sim \mathit{Bernoulli}(p_{c, z})$.
    \end{enumerate}
\end{enumerate}
Here, $\mathit{Dirich}(\mathbf{1}^K)$ is the Dirichlet distribution with unit concentration parameters, while $\mathit{Cat}(\alpha_1,
\ldots, \alpha_K)$ is the categorical distribution parametrized by components' probabilities. Note that the Dirichlet
prior in our IAM mixture is, again, uninformative: Dirichlet distribution with unit concentrations is the
uniform distribution over the $K-1$ dimensional probability simplex. Using similar prior distributions we can also
formulate Bayesian models for other models (see the full version of the paper).

The flexibility of Bayesian models comes with a price: due to the intractable normalization constant in the Bayes rule, it
is typically impossible to evaluate the posterior distribution exactly. That said, there are efficient, general-purpose
algorithms that can be used to draw samples from the posterior distribution. %
We generate posterior samples
using the No-U-turn sampler~\citep{Hoffman2014} with variable elimination~\citep{Obermeyer2019} for the component
assignments and Gibbs sampling for the central votes. To this end, we implement and estimate our models in the NumPyro
probabilistic programming language~\citep{Phan2019}.

After sampling from the posterior distribution, we approximate posterior means of model's parameters by averaging across
sampled values. We then use these mean estimates in downstream analyses. Note that our models use so-called exchangeable
priors: the model specification and, consequently, the posterior density is invariant to permutation of component
labels. In a naive implementation, samples from such models may differ in the ordering of components, leading to
incorrect mean estimates. We remedy this issue by using a standard \emph{identifiability constraint}
technique~\citep{Jasara2005}. In particular, we restrict the Dirichlet prior on the components' probabilities to the
polytope that satisfy the constraint: $\alpha_1 > \alpha_2 > \cdots > \alpha_K$, and put zero prior probability mass
elsewhere. This constraint uniquely identifies one out of $K!$ equivalent component labellings. In practice, the
constraint can be enforced post sampling, by reordering the components in the samples~\citep{Jasara2005}.

\subsection{Why MLE and  Bayesian Learning}\label{sec:mle_vs_bayes}

We described two different ways to estimate parameters in our mixture
models, MLE and Bayesian learning. Hence, one may ask why estimate the
same model with two different statistical frameworks? Our motivation
for this choice comes from specific strong points of both methods. In
particular, maximizing the data likelihood is a conceptually simple
estimation criterion. It does not require choosing any distribution
other than the components' distributions (and fixing the number of
components). The corresponding EM algorithm is one of the standard
choices for fitting mixture models. As we will see, it also gives us
fairly well estimated mixture models for approval elections.  That
said, extensions of the EM algorithm to novel component distributions
for approval elections require derivation of the needed update
equations. Bayesian learning is more flexible in this respect:
contemporary probabilistic programming frameworks, such as
NumPyro~\citep{Phan2019}, can conveniently express complex generative
processes and provide generic algorithms to estimate them. Priors in
Bayes models may also serve as regularization terms, preventing us,
e.g., from ascribing zero approval probability to a candidate that
happened to have no approvals in the training data. Finally, while we
do not pursue this direction in our work, posterior distributions may,
in principle, provide uncertainty estimates for the inferred
quantities. Nevertheless, Bayesian modelling comes with a conceptually
more elaborate statistical framework. It requires care when choosing
prior distributions, especially in mixture models where label
switching may affect inferences. Finally, estimation of Bayesian
models often relies on dedicated modelling software.

\section{Experiments}

Our experiments focus on learning variants of IAMs, as well as
their mixtures, on elections from the Pabulib
database~\citep{fal-fli-pet-pie-sko-sto-szu-tal:c:pabulib}.
Specifically, we considered all $271$ approval-based Pabulib elections
that include at least $2\ 000$ voters.  For each election $E = (C,V)$
from this set, and each considered algorithm $\calA$ (for a given IAM
variant) we  executed the following procedure $t_\try = 5$ times (each time
using independent coin tosses; for the experiment described in \Cref{sec:dim} we used
$t_\try = 20$):
\begin{enumerate}
\item We formed elections $E_\learn$ and $E_\eval$, where the latter
  consisted of randomly selected $n_\eval = 1000$ votes from $E$, and
  the former consisted of the remaining votes. If $E_\learn$ ended up
  with more than $20\ 000$ voters, then we kept only
  $n_\sample = 20\ 000$ of them, selected uniformly at random (to
  bound the computation time). We refer to $E_\learn$ as the
  \emph{learning} election and to $E_\eval$ as the \emph{evaluation}
  one.
\item We run $\calA$ on $E_\learn$ and obtained distribution
  $D \in \calD(C)$.
\item We computed the log-likelihood of obtaining $E_\eval$ using~$D$,
  as well as a few other metrics (see \Cref{sec:metrics}).
\end{enumerate}
We only performed $5$ runs of this procedure because we found that the
variance of the results that we get (i.e., variance of the metrics
that we computed) was typically several orders of magnitude lower than
the results themselves.

For each of our elections, we ran all the algorithms from
\Cref{sec:learn-single-iam}, i.e.,\ the single-component algorithms for
IC, Hamming, Resampling, full IAM and all the $t$-parameter IAM models for $t$
ranging from~$1$ to the number of candidates in the given election
(with a step of~$1$). Then, we applied Bayesian learning
(\Cref{sec:bayes}) to compute mixture models with $2$, $3$, and $4$
components, where each of the components was either a Hamming model, a
resampling model, or full IAM.  Finally, we applied the EM algorithm
(\Cref{sec:em}) to learn mixture models of $2$, $3$, and $4$ full IAM
components (we omitted Hamming and resampling models due to
computation cost).

\subsection{Evaluation Metrics}\label{sec:metrics}

While we could use log-likelihoods of the models that we learn to
evaluate their quality, this has drawbacks.
For example, it is 
difficult to compare log-likelihood values across different elections. Thus, we
use metrics based on the voter-anonymous variant of the Hamming
distance, defined below.

\begin{definition}
  Let $E = (C,V)$ and $F = (C,U)$ be two elections over the
  same candidate set $C$, with voter collections
  $V = (v_1, \ldots, v_n)$ and
  $U = (u_1, \ldots, u_n)$ of equal size. Their
  voter-anonymous Hamming distance is as follows ($S_n$ is the
  set of permutations of $[n]$):
\begin{equation*}
  \textstyle
  \vah(E, F) =\frac{1}{n} \min_{\sigma \in S_n}\sum_{i=1}^\ell \ham\big(A(v_i), A(u_{\sigma(i)})\big).
\end{equation*}
\end{definition}
In other words, $\vah(E,F)$ is the average Hamming distance
between the votes from $E$ and $F$, matched in such a way as to
minimize the final result. Note that our definition is similar to the
definitions of isomorphic distances of
\citet{fal-sko-sli-szu-tal:c:isomorphism} and
\citet{szu-fal-jan-lac-sli-sor-tal:c:sampling-approval-elections},
except that we consider election with equal candidate sets.
In particular, voter-anonymous Hamming
distance is invariant to reordering the voters and is normalized by
the number of voters.

\paragraph{Baseline Distance}
Let $E$ be an election from the subset of Pabulib that we consider. We
define $E$'s baseline distance as the expected value of the random
variable defined as $\vah(E_1,E_2)$, where $E_1$ and $E_2$ are
subelections of $E$, each with $n_\eval$ voters, selected uniformly at
random up to the condition that $E_1$ and $E_2$ do not have any voters
in common.\footnote{E.g., it means that if some voter $v$ from
  $E$ is included in $E_1$ then he or she is certainly not included in
  $E_2$. However, $E_2$ may contain other voter $u$ with
  $A(v) = A(u)$.} Intuitively, baseline distance is a measure of an
election's internal diversity. For example, if its value is $2$ then
if we take two random, disjoint subelections of $E$ (each with
$n_\eval$ voters), it would be possible, on average, to match their
votes so that two matched votes differ on two candidates
(e.g., each of them may include a single candidate not present
in the other one).
In practice, we compute baseline distance of an
election by drawing $5$ pairs of elections and averaging their
voter-anonymous Hamming distance (typically, the variance is orders of
magnitude lower than the value of the average, so considering $5$
pairs of elections %
is justified).

\paragraph{Absolute and Relative Distances}
Given a Pabulib election $E = (C,V)$ and a learning algorithm $\calA$,
we compute their \emph{absolute distance} as follows: For each
evaluation election $E_{\eval}$ that we computed for $E$ and $\calA$,
we take distribution $D \in \calD(C)$ obtained by $\calA$ on
$E_\learn$, generate election $E_D$ by drawing $n_\eval$ votes
independently from $D$, and compute $\vah(E_\eval, E_D)$. We obtain
five numbers and we output their average value.  We define the
\emph{relative distance} between $E$ and $\calA$ as their absolute
distance divided by $E$'s baseline.  In other words, relative distance
normalizes the absolute one by $E$'s inherent diversity %
($E$'s baseline is, essentially, %
its absolute distance from a distribution that samples $E$'s votes
uniformly at random so, intuitively, it bounds %
achievable absolute distance).

\begin{figure}
  \centering
  \includegraphics[width=5.2cm]{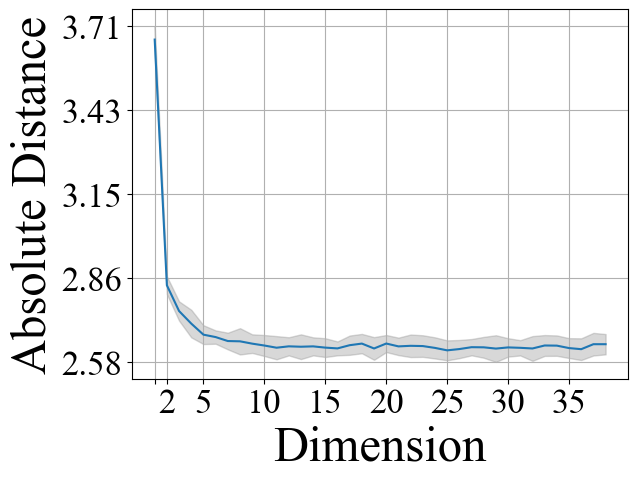}~\qquad
  \includegraphics[width=5.2cm]{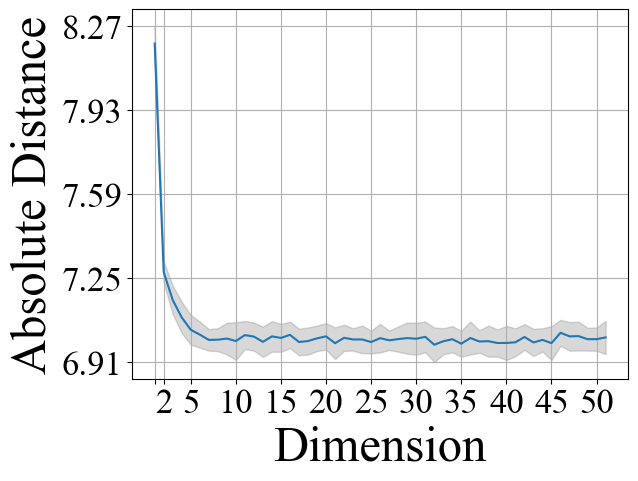}%
  \caption{\label{fig:dim}Absolute distance between the Amsterdam 289
    election (38 candidates, left) or Warszawa 2020 Ochota election
    (51 candidates, right) and single-component
    $t$-parameter  IAMs, as a function of
    $t$, from $1$ to the number of candidates.}
\end{figure}

\subsection{Impact of the Number of IAM Parameters}\label{sec:dim}

Let us now focus on single-component IAMs and the influence that the
number of parameters has on their ability to learn Pabulib elections.
In particular, in~\Cref{fig:dim} we plot the absolute distance between
elections generated using $t$-parameter IAM models learned (on two
example Pabulib elections) using the algorithm from
\Cref{thm:t-iam-algo}.  We find that for IC ($t=1$) we get a
significantly higher absolute distance than for the resampling model
($t=2$), which itself is somewhat higher than the absolute distance
for full-IAM ($t$ equal to the number of candidates).
The plots for other elections are very similar in spirit.

Our conclusion from this experiment is that impartial culture performs
notably worse than the other IAM variants, but models with two parameters
and more achieve fairly similar results (even if
there is still a visible difference between the results for the
resampling model and full IAM). In the following experiments we %
limit our attention to the Hamming, resampling, and full IAM models,
as they are simple to learn, give good results, and the former two can
be specified using much less information than full IAMs.

\subsection{General Analysis of Learning Results}\label{sec:general-analysis}

\begin{figure}[t!]
     \centering
     \includegraphics[width=6.2cm]{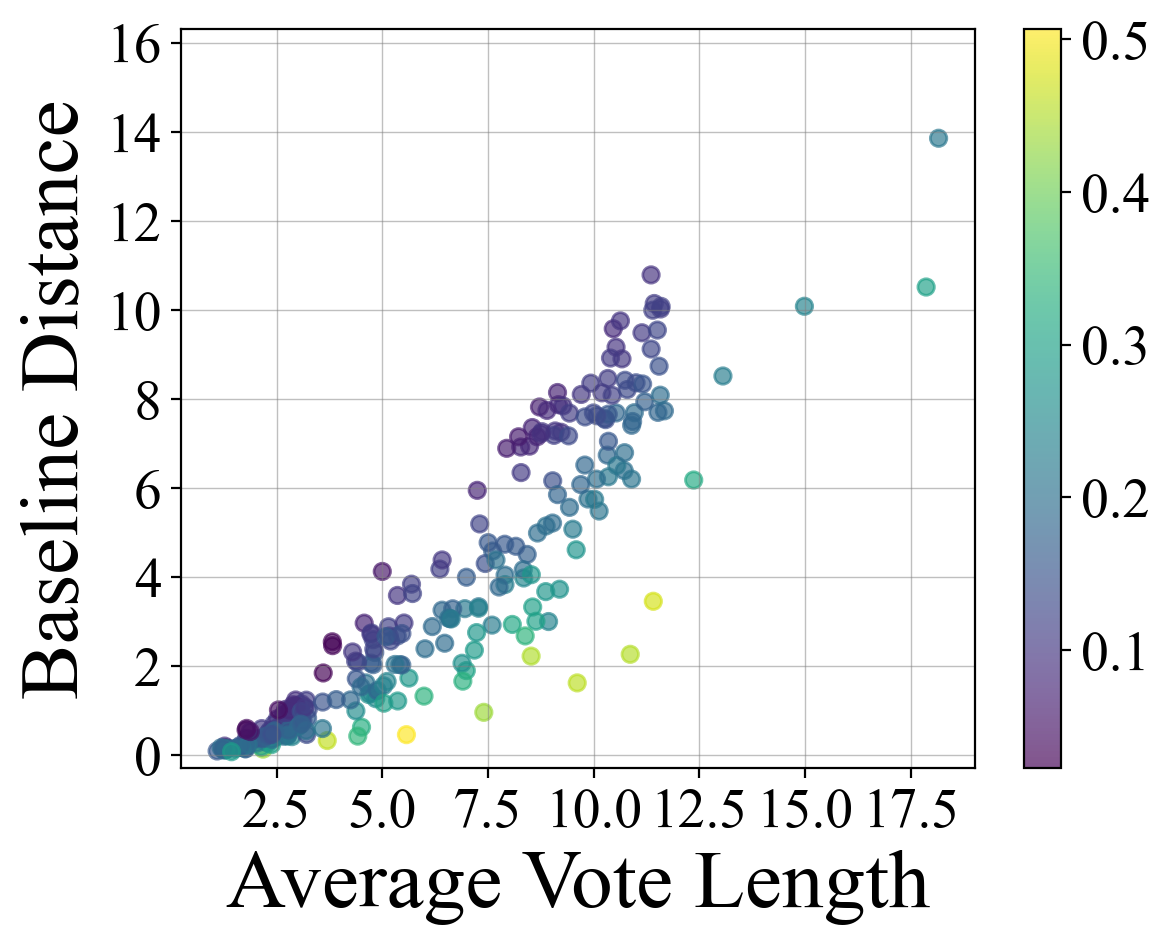}%
     \includegraphics[width=6.2cm]{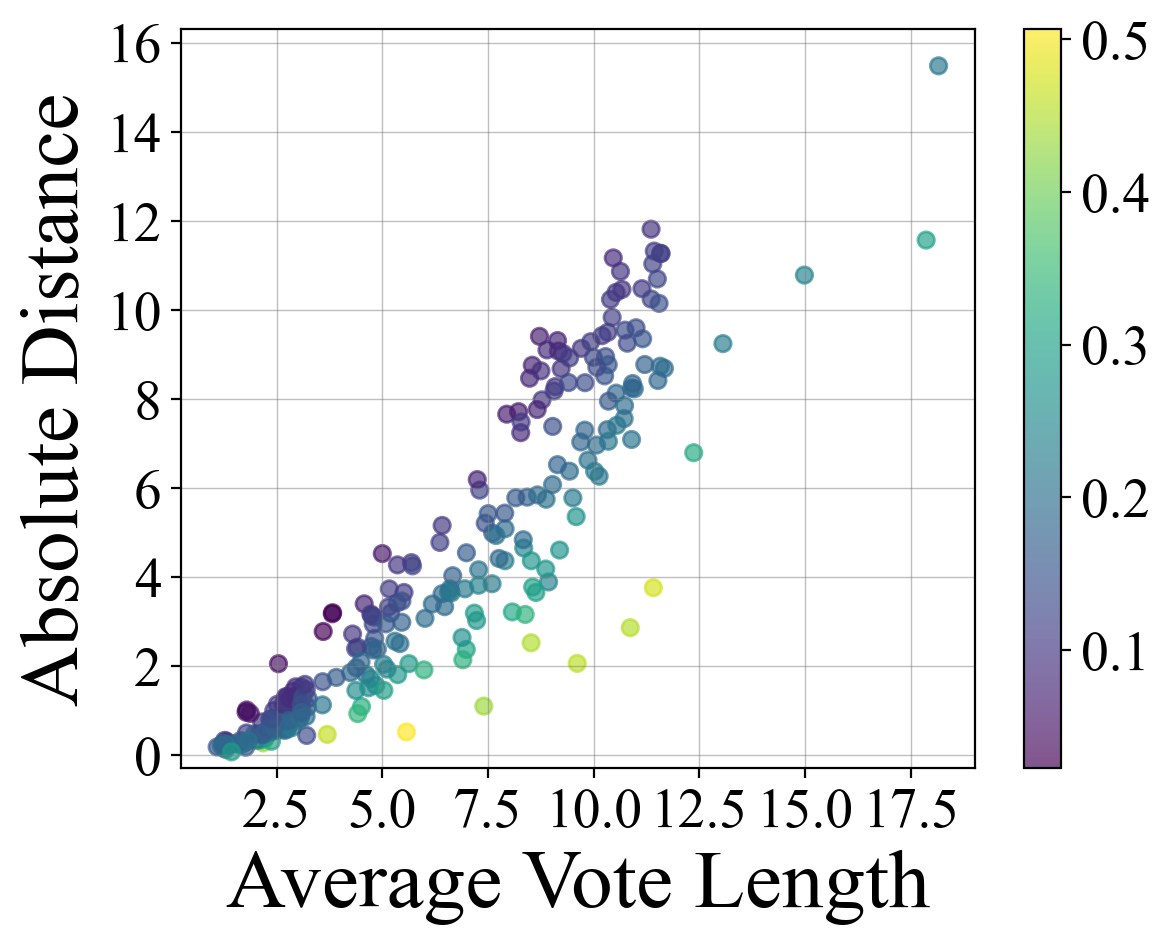}%
     \caption{\label{fig:avg_vote_len_vs_ham_dist}The relation between
       the average vote length and baseline distance (left plot), and
       absolute distances from the best learned model (right
       plot). Each dot depicts a single Pabulib instance. The color
       gives the profile's saturation (i.e., the average vote length
       divided by the number of candidates).}
\end{figure}

\begin{figure}[t!]
     \centering
     \includegraphics[width=6.2cm]{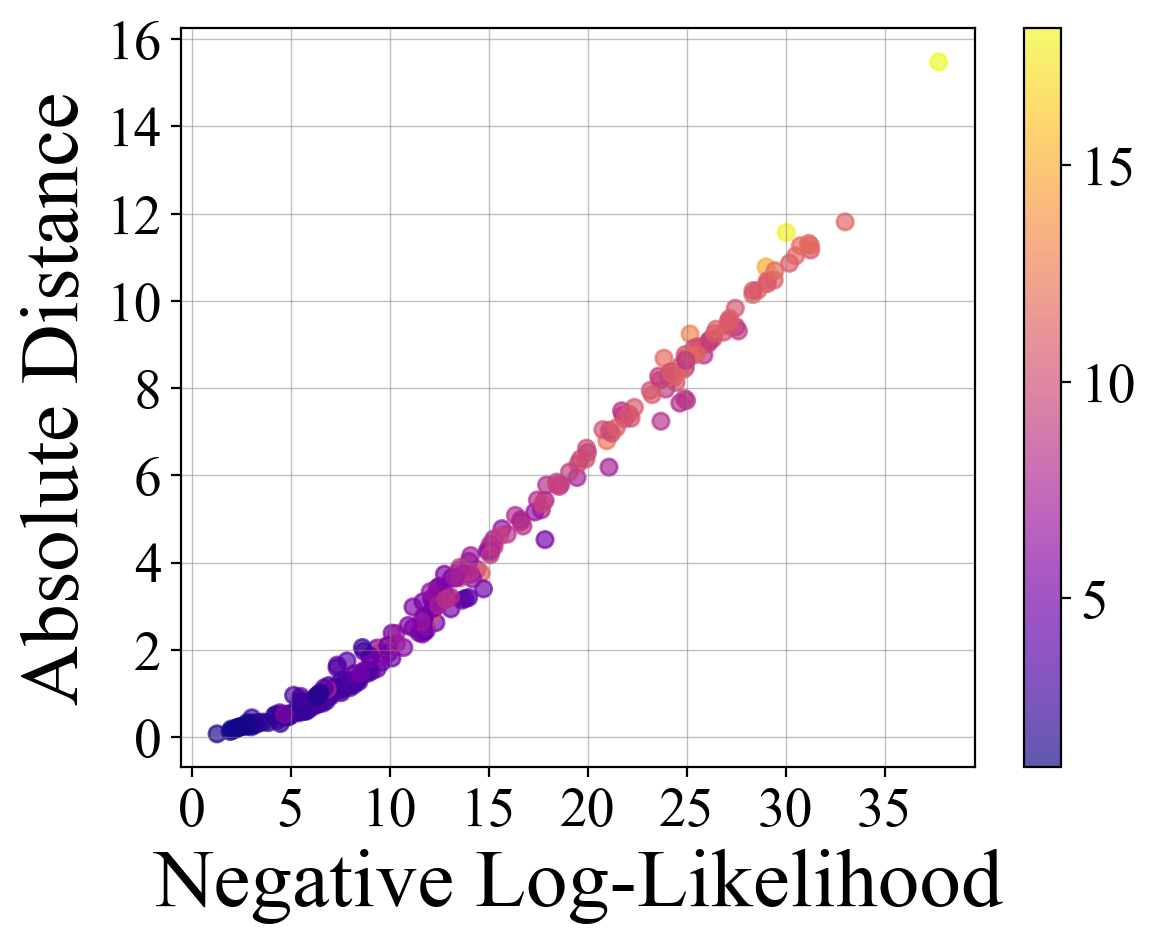}%
     \includegraphics[width=6.2cm]{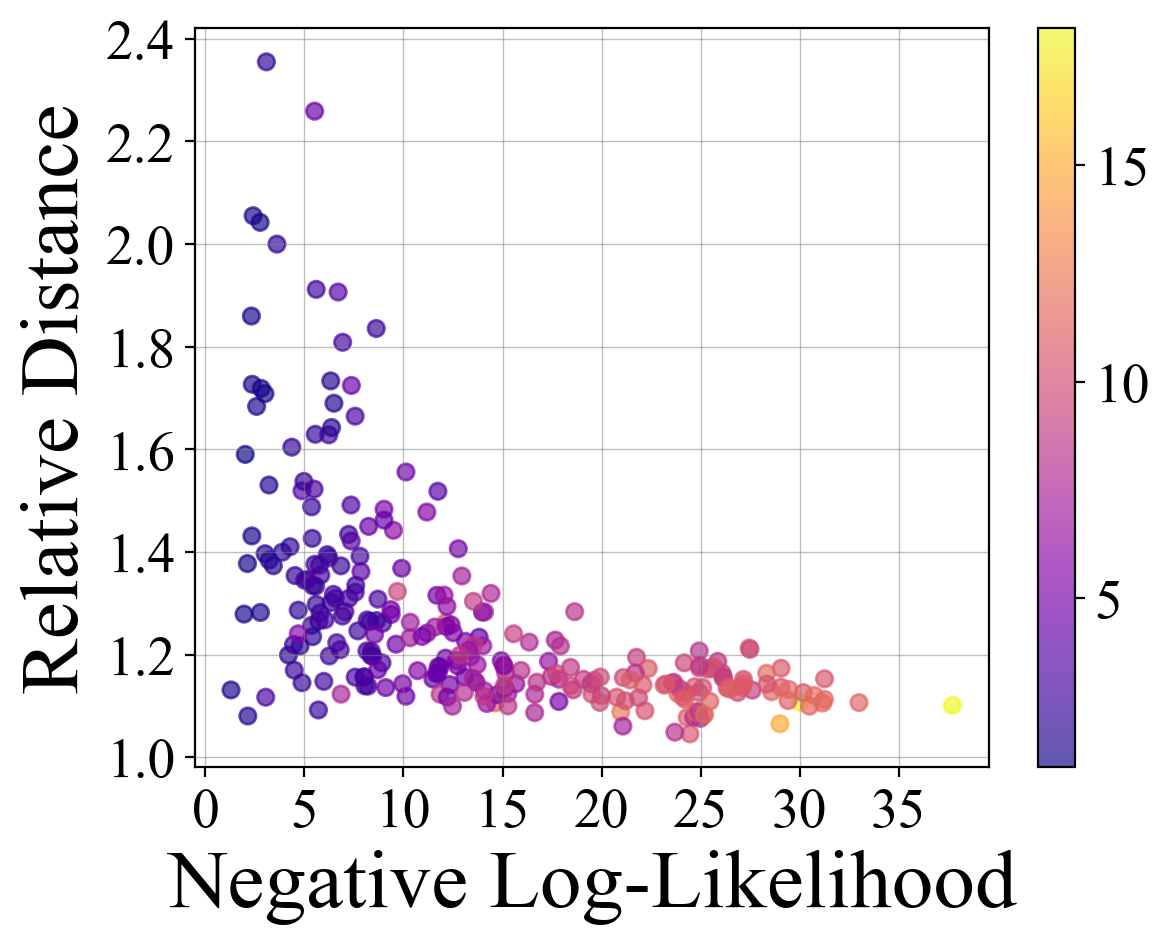}
    \caption{\label{fig:ll_vs_ham_dist}The relation between the (negative) log-likelihood and the absolute distance (left plot), and relative distance (right plot). Each dot depicts a single Pabulib instance. The color corresponds to the average vote length.}
\end{figure}

\begin{figure}[t!]
     \centering
     \includegraphics[width=6.3cm]{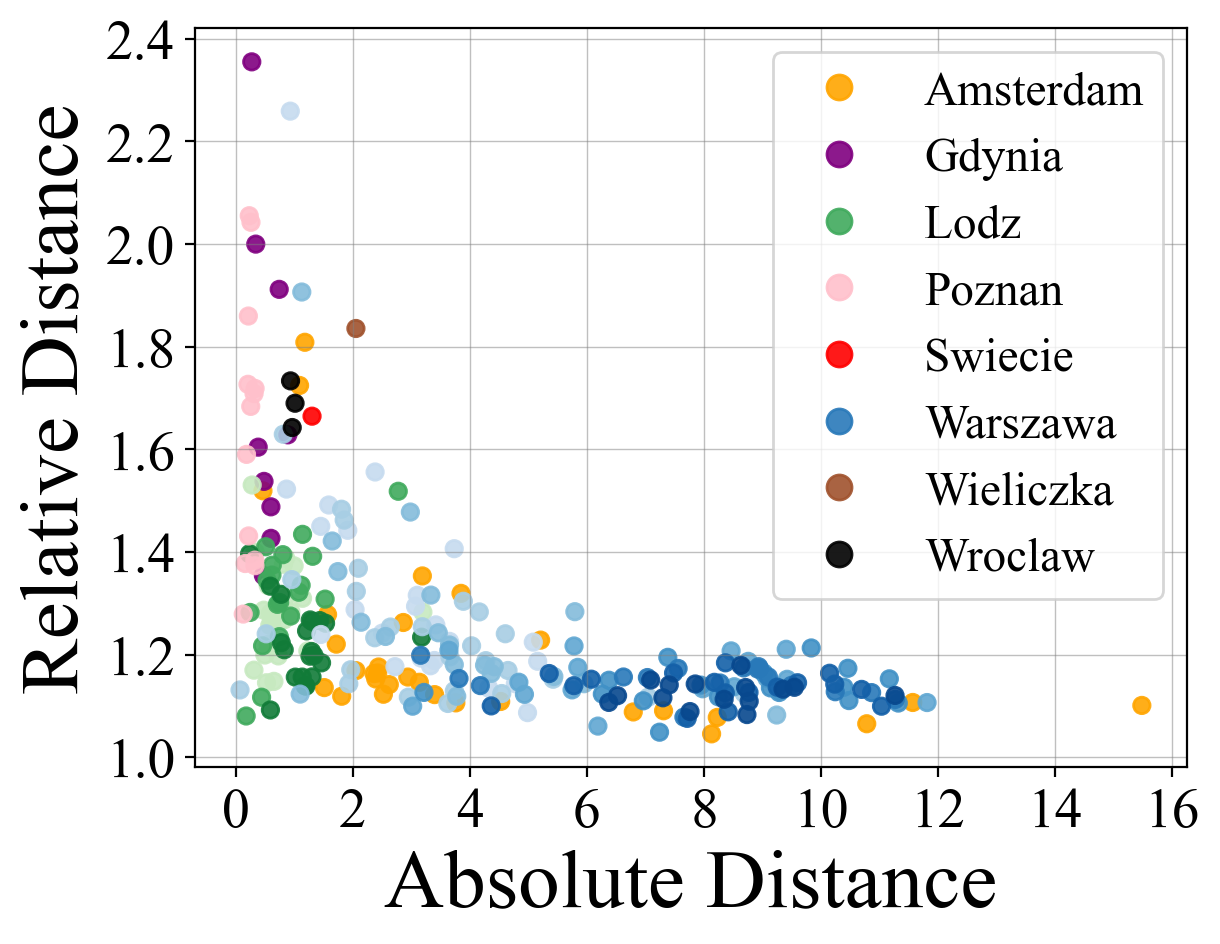}%
     \includegraphics[width=6.3cm]{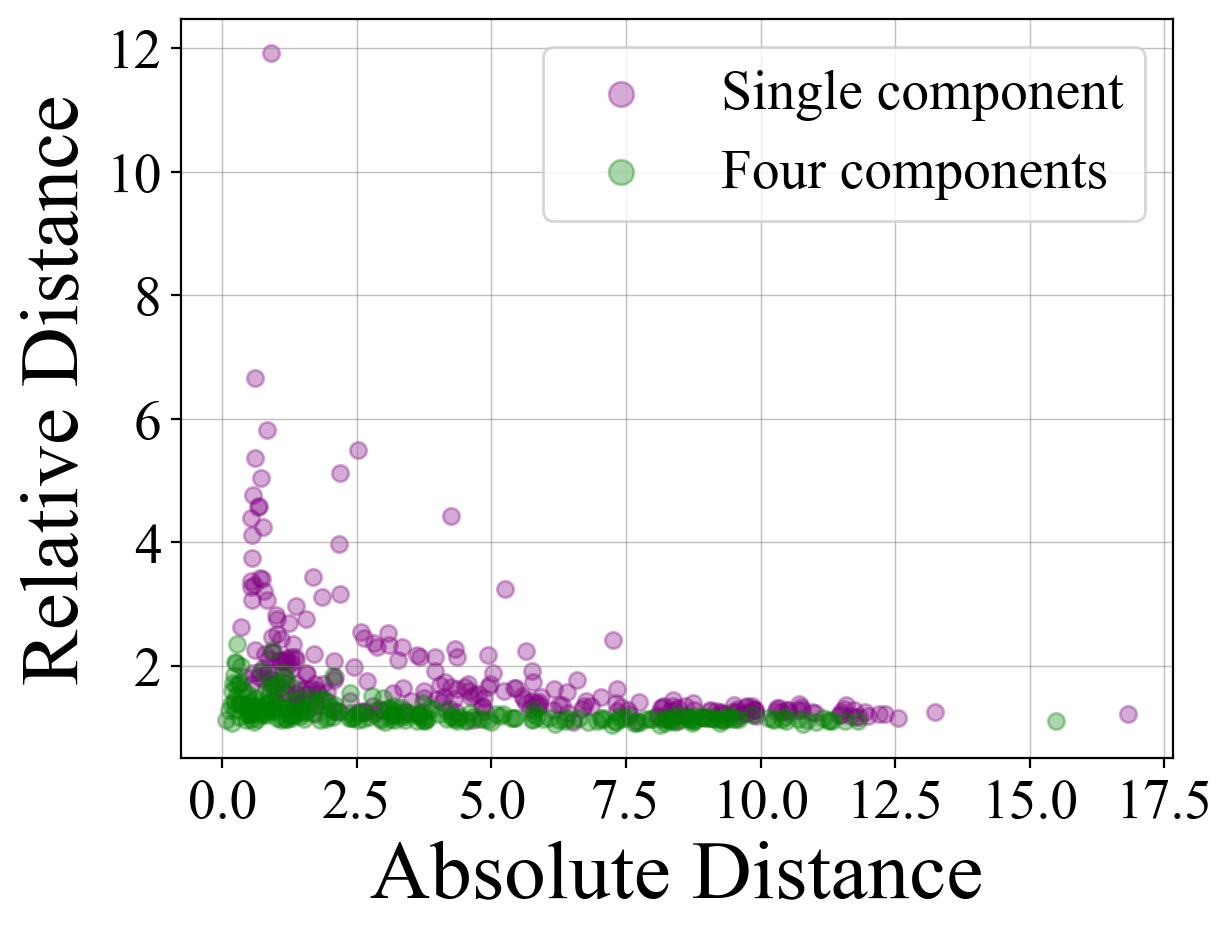}
    \caption{\label{fig:city_and_components}The comparison of the absolute and relative distances.
    The left plot shows from which city each instance originates
    (for Lodz and Warszawa we use different shades of green and blue, respectively, for
    different years).
    The right plot compares the single- and multi-component approaches. Each dot depicts a single Pabulib instance
    (the number of dots doubles in the right plot due to two approaches used).}
\end{figure}

\begin{figure}[t!]
     \centering
     \includegraphics[width=4.9cm]{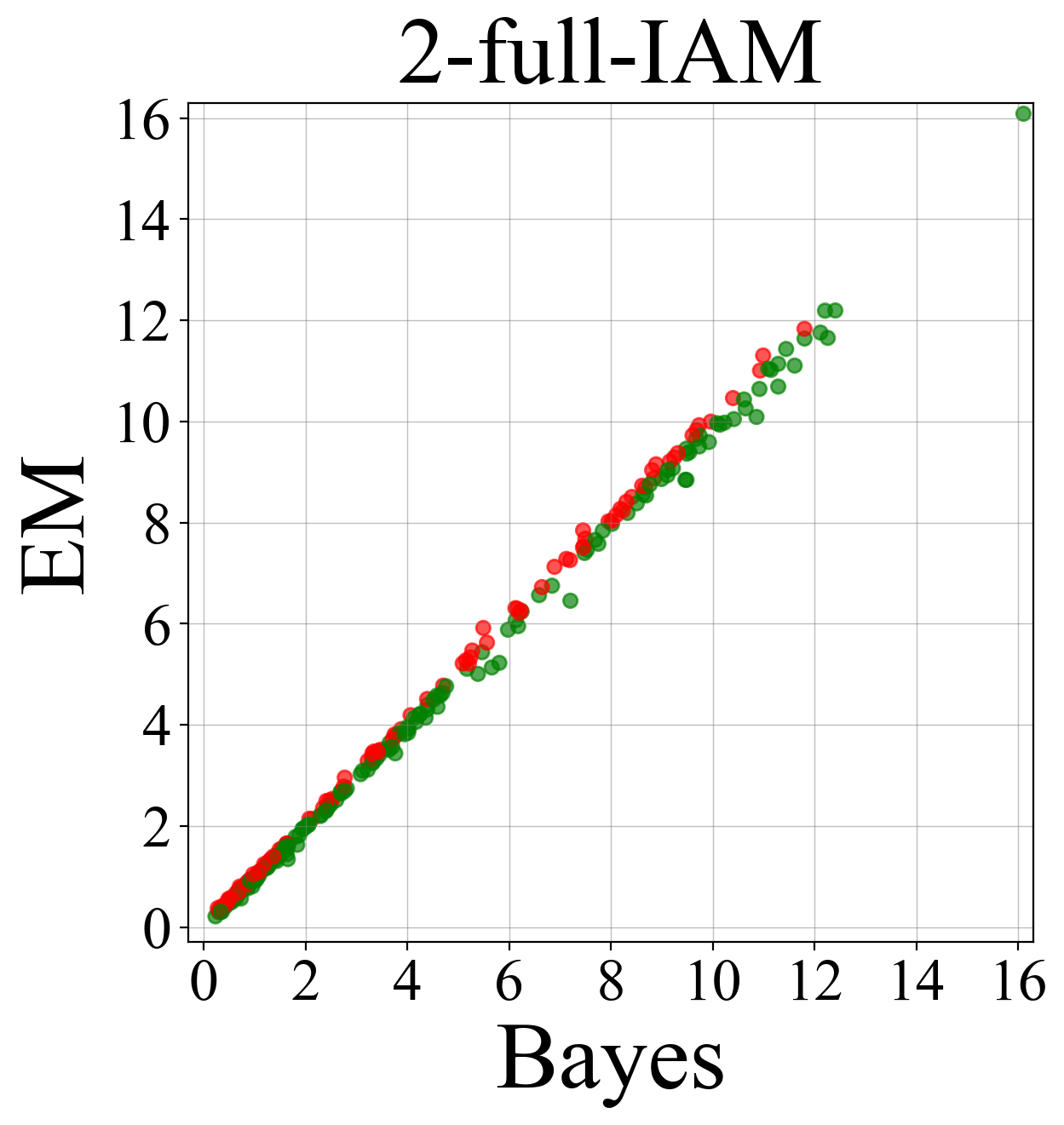}%
     \includegraphics[width=4.9cm]{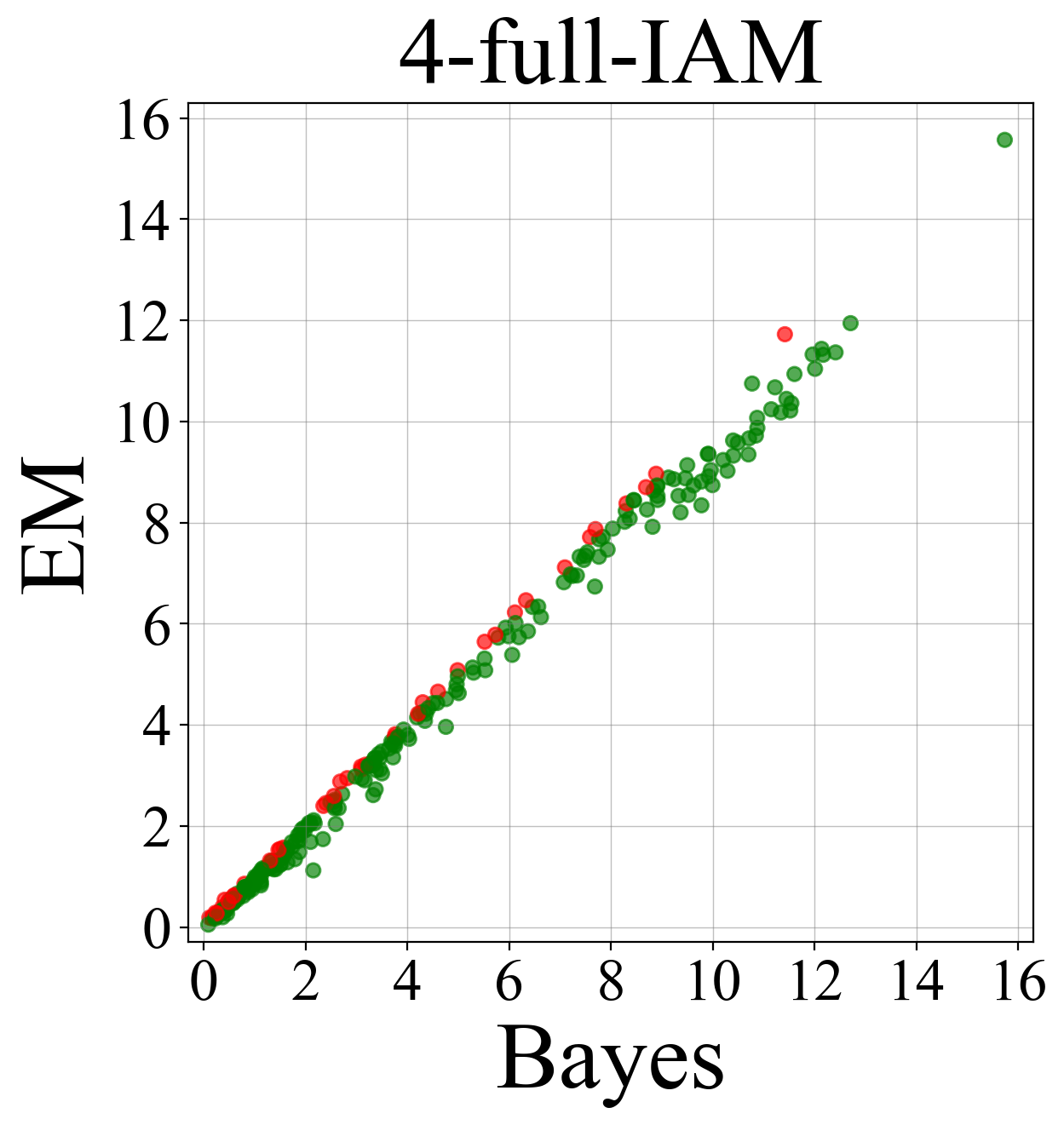}\\
     \includegraphics[width=4.9cm]{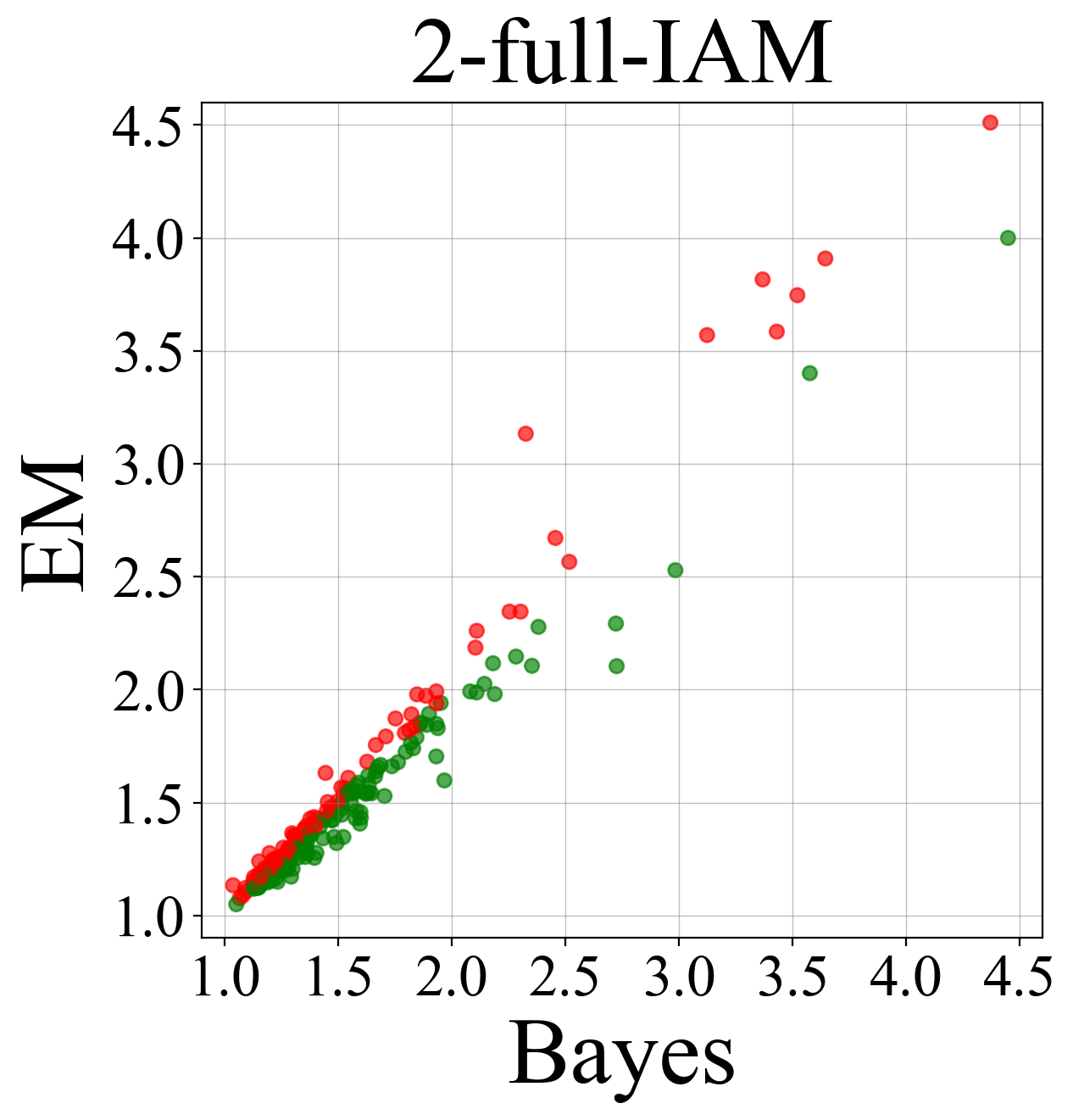}%
     \includegraphics[width=4.9cm]{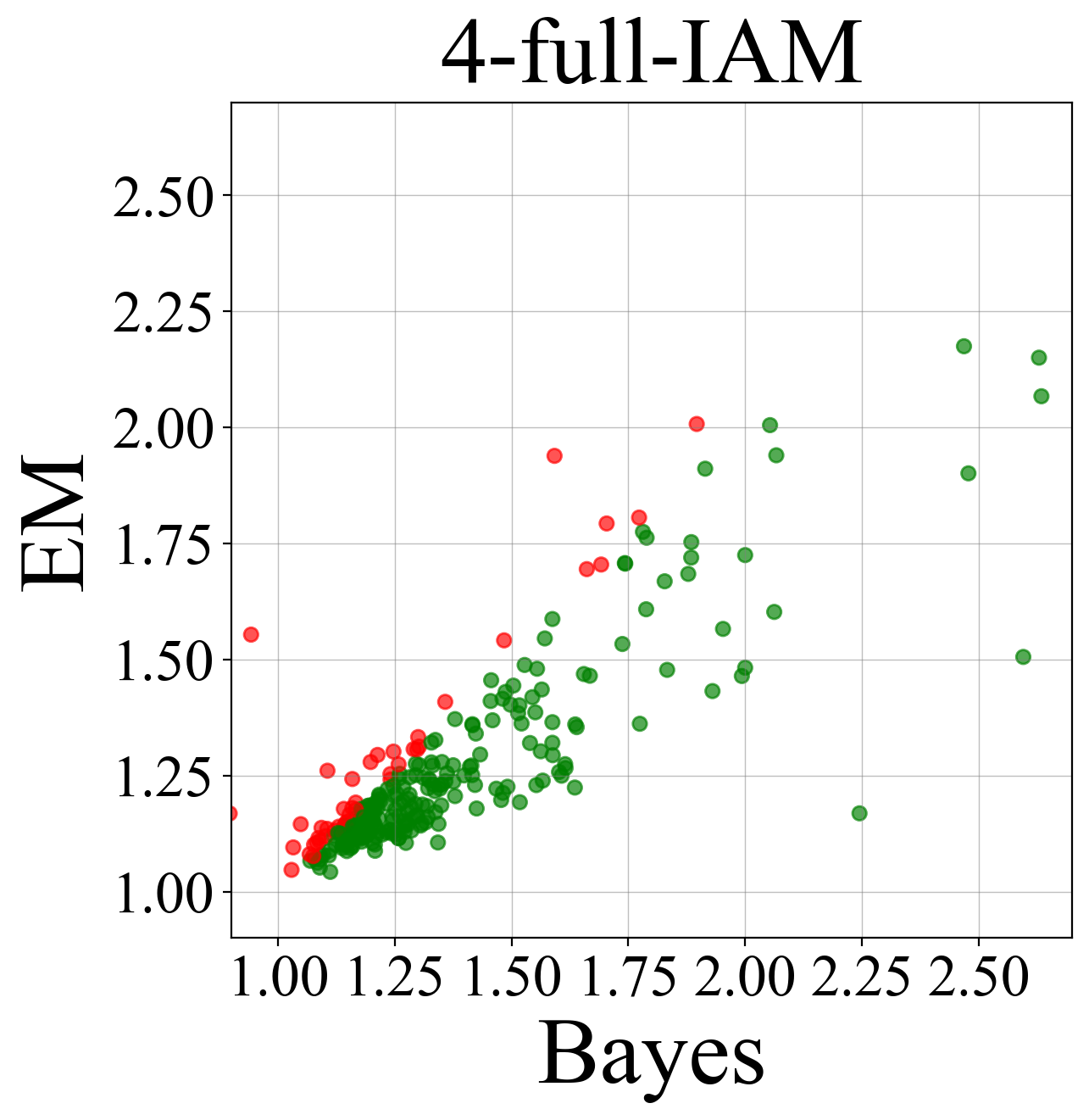}
     \caption{\label{fig:bayes_vs_em}Comparison of Bayes and EM
       learning methods. The upper plots show the absolute
       distances for both methods, and the lower ones show the
       relative distances.  By red (green) color we mark the cases
       where Bayes (EM) achieved smaller distance. Each
       dot depicts a single Pabulib instance.}
\end{figure}

In this section, we provide a high-level overview of the Pabulib
dataset and the performance of our learning
methods. In~\Cref{fig:avg_vote_len_vs_ham_dist} we illustrate the
relationship between the average vote length in a given election
(i.e., the average number of candidates approved by a single voter)
and the baseline distance of this election (left plot) or its absolute
distance from the best-learned model (right plot). Each dot represents
a single Pabulib instance, with the color corresponding to the profile
saturation (i.e., the average vote length divided by the total number
of candidates). The left plot reflects the self-similarity of the
Pabulib instances. We see that while some Pabulib elections have low
baseline distance (i.e., are close to the $x$ axis) for many of them
this is not the case. Indeed, the baseline distance seems to have a
close-to-linear dependence on the average vote length; the more
candidates the voters approve, the more diverse are the votes that
they cast.  The close resemblance between the two plots in
\Cref{fig:avg_vote_len_vs_ham_dist} indicates that, in most cases, our
learning algorithms perform well and achieve absolute distance similar
to the baseline one (we discuss this in more detail later, when
analyzing \Cref{fig:city_and_components}). Regarding the plot on the
right-hand side, we observe two key trends: First, as the average vote
length increases, the absolute distance also increases. Second, for a
given average vote length, higher saturation tends to correspond to a
smaller absolute distance.

In~\Cref{fig:ll_vs_ham_dist} we explore the relationship between the
(negative) log-likelihood and the absolute distance (left plot) and
the relative distance (right plot), both for the best-learned model
(i.e., the one that achieves the lowest absolute distance). While the
log-likelihood is strongly correlated with the absolute distance
(having Pearson correlation coefficient equal to 0.993), it is barely
(negatively) correlated with the relative distance (having
PCC=-0.567). Our conclusion here is that by using distance-based
metrics of quality we gain interpretability of our results (as
discussed in \Cref{sec:metrics}) without losing much of statistical
significance of log-likelihoods (as absolute distances are, in
essence, negative log-likelihoods in disguise).

\Cref{fig:city_and_components} compares the absolute and relative
distances between each election and its best-learned model (i.e., the
one that achieves lowest absolute distance). The left plot uses color
to show the city of origin for each instance, revealing that different
cities tend to occupy distinct regions of the plot. This suggests
significant variation in the nature of elections across cities and is
a strong argument to use data with different origins in experiments
based on Pabulib.

We note that absolute and relative distances tell us quite different
stories about the quality of a learned model. For instance, we may view
relative distance below $1.2$ as quite good, but it may still
correspond to the absolute distance of, say, $10$.  For the case
where, on average, each voter in the considered election also approves
$10$ candidates, this means that our algorithm learned a very good
model as compared to the baseline distance, but the input election is
internally so diverse that the generated votes will still largely
differ from those present in the actual data.
Similarly, we may view relative distance equal to $2$, as rather
unsatisfactory, but it might still mean an absolute distance of $0.5$
for the baseline of $0.25$. Even though the relative distance is
large, the absolute one is objectively small and the learned
distribution produces votes that can be seen as very similar to those
present in the considered election.
Finally, there also are some elections for which both
distances are quite high.
Then, we have to
concede that either IAM mixtures, or our algorithms, are simply
insufficient to learn these elections well. 

Let us now consider the right-hand plot of
\Cref{fig:city_and_components} (note different scales on the axes as
compared to the left-hand one).  This plot contrasts best-learned
single-component models (in this case these always are the full-IAM
ones) and best-learned mixture models (these often are $4$-full-IAMs
learned using the EM algorithm, but sometimes also $4$-resampling or
$4$-full-IAM ones learned using the Bayesian approach). As expected,
mixture models perform much better: Their points have lower absolute
and relative distance coordinates.  While the fact that mixture models
are more expressive than single-component ones is hardly surprising,
knowing the extent of their advantage is useful. Indeed, we see that
to generate realistic elections similar to those in
Pabulib, using mixture models of several IAMs
gives notably
better results that using single-component full-IAMs (not to mention
even simpler single-component models).

In~\Cref{fig:bayes_vs_em} we compare the EM and Bayesian
approaches. As the plots show, particularly in the case of the
4-full-IAM model, the EM approach consistently outperforms the
Bayesian method. Indeed, in this case the Bayesian approach often
finds it difficult to identify four components and outputs models that
perform even worse than the $3$-full-IAM models that it identifies.
While we believe that one could improve on this by engineering the
priors used in Bayesian models, we did not pursue this direction and
view it as a follow-up work.
The main conclusion from \Cref{fig:bayes_vs_em} is that both EM and
the Bayesian approach achieve similar results, even though the former
explicitly minimizes the negative log-likelihood and the latter does
not. This reinforces our view that both algorithms---based on so
different principles---identify meaningful components. As component
analysis can be challenging, we find this finding~valuable.

\subsection{Closer Look on a Few Instances}
\begin{figure}[t!]
     \centering
     \includegraphics[width=5.2cm]{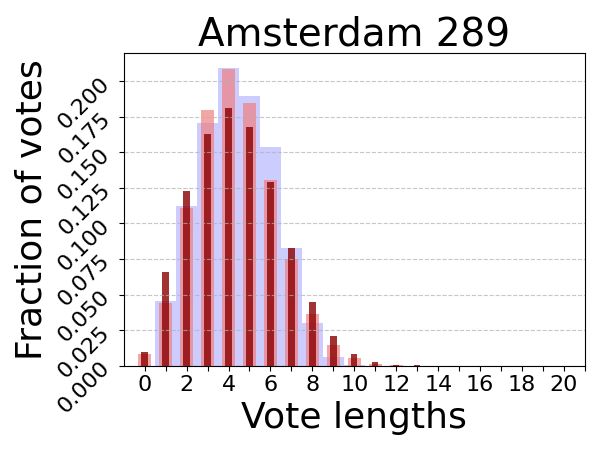}%
     \includegraphics[width=5.2cm]{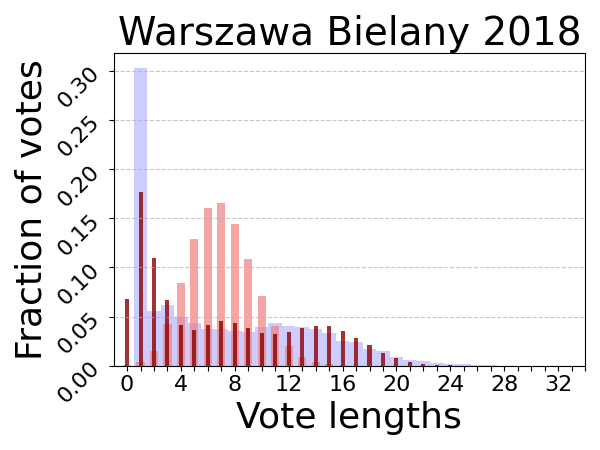}%
     \includegraphics[width=5.2cm]{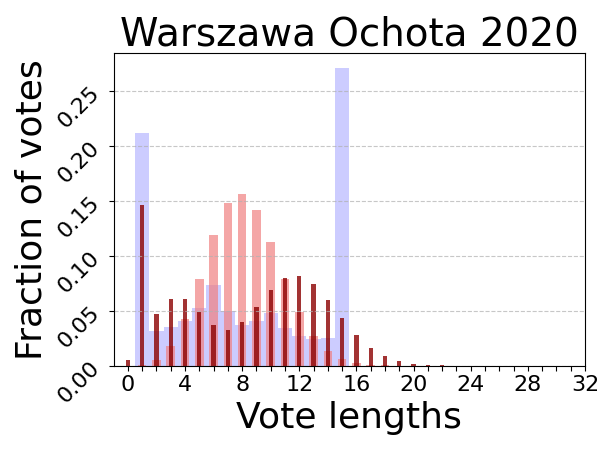}%
     \caption{\label{fig:superimposed} Superposition of three
       histograms of vote lengths for three Pabulib instances. Each
       picture shows histograms of the training election (blue) and
       learned single-component (pale red) and four-component (dark
       red) resampling models. For clarity, we
       removed bars for
       fractions
       smaller than~$10^{-10}$.}
\end{figure}

Next, we zoom in on a few selected instances to describe our
observations in more detail. To improve understanding, we use
histograms which show the number of votes of a given length in an
election.  In~\Cref{fig:superimposed}, for selected Pabulib elections
we lay over histograms of the respective~$E_\learn$ and of the learned
single-component and $4$-component resampling models (results for
Hamming and full-IAM
are similar, we chose resampling for variety).

As mentioned in~\Cref{sec:general-analysis}, nearly always
multiple-component
models yield elections with significantly smaller distances to the
original one. There are two main reasons to explain this observation:
First, single-component models are prone to ``the flaw of average'' of
the vote lengths. It is evident for elections with bimodal
vote length distributions, like that of election Warszawa Ochota
2020 in~\Cref{fig:superimposed}: A
single learned component mostly generates average-length votes,
which are dissimilar from those %
of either of the peaks.

Second, single-component models can only produce votes with a
relatively limited variance of their length. Hence, it is frequently
counterproductive to apply single-component models to learn elections
where this variance is high. It includes elections with a vote length
distribution that is asymmetric, uni-modal, and has a ``heavy tail''
(such as Warszawa Bielany 2018 depicted in~\Cref{fig:superimposed}) or
that have a bimodal distribution of vote lengths (like Warszawa Ochota
2020 in~\Cref{fig:superimposed}).  The histograms clearly show how
multiple components help in dealing with such distributions.

On the positive side, we want to stress that single-component models can
perform well on some real-life elections. In particular, this holds
true for elections with Gaussian-like distributions of vote lengths;
for an example, see election Amsterdam 289 in~\Cref{fig:superimposed}.

\section{Summary}
To the best of our knowledge, we performed the first comprehensive
analysis of Pabulib elections by learning them using both
single-component and mixture models. We found that some elections can
be captured with simple models such as resampling or full IAM, but
typically using mixtures of a few components is preferable.

\section*{Acknowledgments}
  This project has received funding from the European Research Council
  (ERC) under the European Union’s Horizon 2020 research and
  innovation programme (grant agreement No 101002854), from the French
  government under the management of Agence Nationale de la Recherche
  as part of the France 2030 program, reference ANR-23-IACL-0008.  In
  part, A. Kaczmarczyk acknowledges support from NSF CCF-2303372 and
  ONR N00014-23-1-2802.  In part, S. Szufa was supported by the
  Foundation for Polish Science (FNP).  In part, M. Kurdziel was
  supported by Poland’s National Science Centre (NCN) grant
  no. 2023/49/B/ST6/01458.  We gratefully acknowledge Polish
  high-performance computing infrastructure PLGrid (HPC Center: ACK
  Cyfronet AGH) for providing computer facilities and support within
  computational grant no. PLG/2024/017160.
  \begin{center}
    \includegraphics[width=3cm]{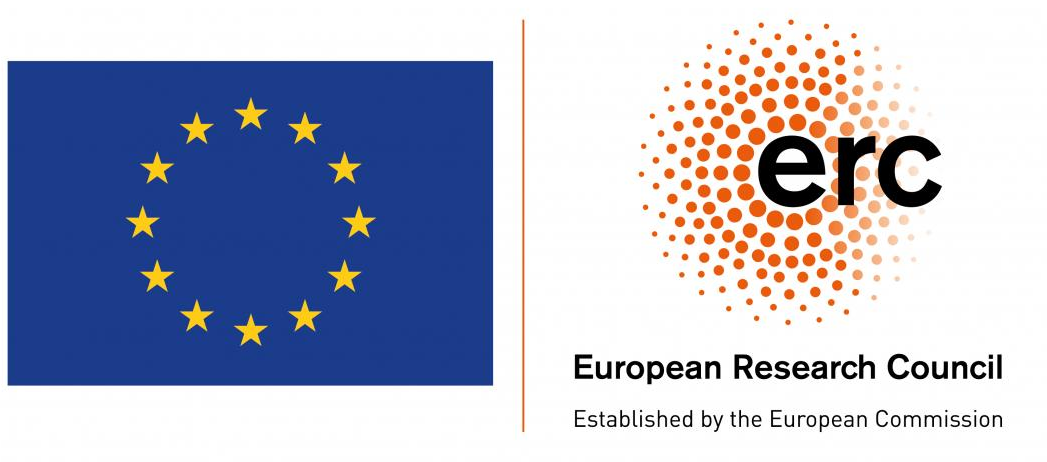}
  \end{center}

\balance{}

\bibliographystyle{plainnat}
\bibliography{bib, bayes}

\appendix

\section{Missing Corollary for Expectation Maximization}\label{apdx:em}

Below we provide a corollary based on
\Cref{prop:maximization-simplification} that obtains analogues of
theorems presented in \Cref{sec:learn-single-iam} for IC, full IAM,
Hamming and resampling models. They differ from their original
versions so that for each $k$-th mixture component the number of
voters $n$ is replaced by $Q\cdot \gamma_k$ and each vote $v$ is
multiplied $Q\cdot \gamma_{v,k}$ times---the value $Q$ in all cases
can be actually reduced from the formulas.

\begin{corollary}
  Let $E = (C,V)$ be an approval election with $m$ candidates and $n$
  voters, such that
  $|V(c_1)| \geq |V(c_2)| \geq \cdots \geq |V(c_m)|$, and let $\gamma$ be the soft assignment of votes to some $K$
  IAM mixture components. Then for each $k \in [K]$, the expected log-likelihood
  of generating
  the part of assignment corresponding to the $k$-th component is
  maximized:
  \begin{enumerate}
    \item for $p$-IC, when $p = \nicefrac{1}{\gamma_k m}\cdot \sum_{v\in V}
      \gamma_{v, k} \cdot |A(v)|$,
    \item for full $(p_1, \ldots, p_m)$-IAM, when $p_j = \nicefrac{\sum_{v\in
      V(c_j)} \gamma_{v,k}}{\gamma_k}$ for each $j\in [m]$,
    \item for $(\phi, u)$-Hamming, when $u$ is the
      majoritarian central vote and $\phi = \nicefrac{h}{m\gamma_k -h}$, where $h = \sum_{i=1}^n{\gamma_{v,i}\ham(u,v_i)}$
    \item for $(p_1, p_2)$-IAM, when:
    \begin{align*}
  p_1 &= \frac{\sum_{v\in V(c_1)} \gamma_{v, k}+ \cdots + \sum_{v\in V(c_{m'})}
  \gamma_{v, k}}{\gamma_k m'}, \quad C_1 = \{c_1, \ldots, c_{m'}\}, \\ 
  p_2 &= \frac{\sum_{v\in V(c_{m'+1})} \gamma_{v,k} + \cdots + \sum_{v\in
  V(c_{m})} \gamma_{v,k}}{\gamma_k (m-m')}, \quad  C_2 = C \setminus C_1.
\end{align*}
 for some $m'\in [m]$.
  \end{enumerate}
\end{corollary}

\section{Bayesian models}

\subsection{Mixtures of Resampling and Hamming distributions}\label{apdx:bayes_models}

Using techniques introduced in~\Cref{sec:bayes}, we formulate a Bayesian mixture model with $(p,\phi)$-Resampling components:
\begin{enumerate}
\item Sample component probabilities,
  $\left( \alpha_1, \ldots, \alpha_K \right) \sim
  \mathit{Dirich}(\mathbf{1}^K)$.
\item For all $k \in [K]$:
  \begin{enumerate}
  \item Sample the resampling and approval probabilities, $\phi_k \sim U(0, 1)$, $p_k \sim U(0, 1)$.
  \item For all $c \in C$, sample the central vote element, $u_k(c) \sim \mathit{Bernoulli}(p_k)$.     
  \end{enumerate}
  \item For all $v_i \in V$:
    \begin{enumerate}
    \item Sample the component index, $z \sim \mathit{Cat}(\alpha_1, \ldots, \alpha_K)$.
    \item For all $c \in C$, sample the vote outcome:
      \begin{enumerate}
        \item $\sigma \sim U(0, 1)$, 
        \item $v_i(c) = u_z(c)$, if $\sigma > \phi_z$, and
         $v_i(c) \sim \mathit{Bernoulli}(p_z)$, if $\sigma \leq \phi_z$.
      \end{enumerate}
    \end{enumerate}
\end{enumerate}
Similarly, we formulate a Bayesian mixture model with $\phi$-Hamming components:
\begin{enumerate}
\item Sample component probabilities,
  $\left( \alpha_1, \ldots, \alpha_K \right) \sim
  \mathit{Dirich}(\mathbf{1}^K)$.
\item For all $k \in [K]$:
  \begin{enumerate}
  \item Sample the noise parameter, $\phi_k \sim U(0, 1)$. 
  \item For all $c \in C$, sample the central vote element,  $u_k(c) \sim \mathit{Bernoulli}(\nicefrac{1}{2})$. 
  \end{enumerate}
  \item For all $v_i \in V$:
    \begin{enumerate}
    \item Sample the component index, $z \sim \mathit{Cat}(\alpha_1, \ldots, \alpha_K)$.
    \item For all $c \in C$, sample the vote outcome:
      \begin{enumerate}
        \item[] $v_i(c) \sim \mathit{Bernoulli}(\nicefrac{1}{1 + \phi_z})$ if $u_z(c) = 1$, and
        \item[]  $v_i(c) \sim \mathit{Bernoulli}(\nicefrac{\phi}{1 + \phi_z})$  if  $u_z(c) \neq 1$.
      \end{enumerate}
    \end{enumerate}
\end{enumerate}
In this model we sample from components using equivalence of the
$\phi$-Hamming distribution with the 2-dimensional IAM.

\subsection{Estimation}

For each evaluated model, we use the No-U-turn sampler~\citep{Hoffman2014} implementation provided by the NumPyro~\citep{Phan2019} library to generate~$2000$ samples from the posterior distribution. Following standard practice, we discard initial part of the sampler's trajectory, where the Markov chain might be transitioning to an equilibrium distribution. Specifically, we discard initial~$1000$ samples, and use the remaining~$1000$ to estimate mean parameter values.

\end{document}